\documentclass[11pt]{article}
\usepackage{mystyle}
\usepackage[affil-sl]{authblk}
\usepackage[margin=1in]{geometry}
\allowdisplaybreaks[4]
\title{Thermodynamics of bipartite entanglement.}
\author{M.Sc. Thesis\\by\\Nikolaos K. Kollas}
\affil{Department of Theoretical Physics, University of Patras, Greece}
\date{}
\begin{document} 	
\maketitle
\abstract{A review is given on the thermodynamical structure of bipartite entanglement. By comparing it to the axiomatic formulation of thermodynamics presented by Giles it is shown that for finite dimensional systems the two theories are formally inequivalent. The same approach is used to demonstrate the full equivalence in the asymptotic limit for pure quantum states. For mixed states a different method for obtaining the second law is described applied to two different classes of operations, PPT-preserving and asymptotically non-entangling operations.  
}
\tableofcontents
\section{Introduction.}\label{chap1}
\pagenumbering{arabic}
\setcounter{page}{1}
Quantum mechanics is full of counter-intuitive notions, from Heisenberg's uncertainty principle to most recently the quantum-pigeonhole principle \cite{Aharonov2014}. Another example, which sparked great debate between prominent physicists such as Einstein and Schr{\"o}dinger and the proponents of quantum theory, is the notion of \emph{entanglement}. Entanglement is a property of composite quantum systems defined as the impossibility to separately determine the quantum state of it's parts, as a consequence it's components become `connected' and the results of individual measurements carried out on each component are correlated. This fact was the crux of the famous Einstein-Podolsky-Rosen (EPR) paradox \cite{Einstein1935a}, simplified later by Bohm \cite{Bohm1957} by considering a system in the singlet state 
\begin{labeq}{eq1.1}
\psi_s=\frac{1}{\sqrt{2}}(\ket{01}-\ket{10})
\end{labeq}
this is the canonical example of an \emph{entangled state}, physically it  can be realized as a system of two spin-$\frac{1}{2}$ particles (0 for spin directed on the positive axis of a given direction, 1 if it is directed on the negative axis). It is easy to see from \cref{eq1.1} that when the spin of one particle is measured along a direction, the spin of the second particle will always be found to have an opposite orientation, the main point of the argument is that this happens instantaneously. Suppose now that each particle is sent to two physicists located in different labs separated by a spacelike distance. Upon arrival each particle's spin is subsequently measured along a predefined direction. How is it possible then for the second particle to orientate it's spin accordingly depending on what the result of the measurement on the first one was when according to the special theory of relativity the positions of the two particles are causally unrelated? Einstein's answer was  it is not, the only alternative according to him is that the properties of each particle must be predetermined before separation, quantum mechanics is therefore not a complete physical theory but rather an approximation to an underlying  theory whose variables although unknown to us uniquely determine the properties of the particles, such theories are commonly known as \emph{hidden variable theories}.

It took almost thirty years to settle this argument when in 1964 Bell published a paper providing a means to test experimentally if such theories are physically possible \cite{Bell1964a}. This is done by checking if the expectation values for a specific setting of measurements of spin directions of the singlet state violates or not an elegant inequality named after Bell. A violation would be consistent with what is predicted by quantum mechanics but not by any hidden variable theory. This inequality was later generalized to make it more suitable for experiments \cite{Clauser1969}. To date Bell violations have been verified experimentally \cite{Christensen2013} corroborating the belief that quantum mechanics is fundamental to the description of physical reality.

Over the past twenty five years entanglement has also been recognized as a key ingredient in \emph{quantum information theory} allowing tasks such as teleportation \cite{Bennett1993}, dense coding \cite{Bennett1992a}, secure cryptography \cite{Ekert1991} and the reduction of communication complexity \cite{Buhrman2001} to be performed. Entanglement is also the reason behind the efficiency of quantum computers \cite{Society1985}, machines which if realized would dwarf the computational capabilities of even the most powerful classical computers. 

Most of these tasks require the use of singlets in order to be implemented successfully. Unfortunately such states are very hard to maintain due to interactions with their environment, making them a valuable resource much like bits are in classical information theory. It is therefore important to know when the entanglement of a given state is of the singlet type or has been degraded and whether it is possible or not to recover a singlet state in the latter case, this is the subject of the \emph{theory of entanglement}.

Many of the tools used in this theory are of thermodynamic origin such as the von-Neumann entropy or Landauer's principle \cite{Vedral1999a,Plenio2001b,Oppenheim2002}, conversely it has been shown that heat capacity can be used as a method of detecting entanglement \cite{Wiesniak2008}. Furthermore it is known that the amount of singlets left over after a quantum process is always less than the original a fact which is reminiscent of the second law of thermodynamics. It is therefore possible that the two theories are formally equivalent to each other. This would be extremely useful since we would be able to tell which processes in quantum information theory are possible or not by using the analogue of the second law for entanglement.

In this thesis we will present a detailed examination of such an equivalence reviewing available literature on the subject. In \cref{chap2} we review the axiomatic formulation of the classical theory of thermodynamics making it the basis for later discussions. \cref{chap3} contains the mathematical definition of entanglement, a description of the basic set of operations that are used along with some important entanglement processes, we conclude the section by giving the definition of entanglement measures along with some examples. In \cref{chap4} we discuss the differences between entanglement theory and thermodynamics with the help of an important theorem for entanglement processes between pure states. In \cref{chap5} we give a description of three different cases extending the theory to include infinite dimensional systems. In one of these cases entanglement and thermodynamics are proven to be formally equivalent to each other while for the other two the question remains open. \cref{chap6} contains concluding remarks. Proofs of some important theorems and a description of teleportation which were to lengthy to include in the main text can be found in the Appendix.
\section{The axiomatic formulation of Thermodynamics.}\label{chap2}
Perhaps one of the most striking features of the theory of thermodynamics and certainly the one which causes most confusion is the notion of entropy. Hidden inside the second law  it is usually connected to the amount of heat that is exchanged during a process or the amount of order of a given macroscopic state. In any case it possesses a characteristic property: it can never decrease. To formulate this property if $a$ is a state that during a process evolves into state $b$ then the entropy of the initial state $S(a)$ will always be \emph{less than or equal}	 to that of the final state $S(b)$, this property is known as the \emph{second law of thermodynamics}, whenever we refer to this law we shall tacitly assume this formulation.

Other versions of the second law also exist, most textbooks on thermodynamics usually include the following:
\begin{description}
	\item[Clausius:] There exists no process the sole result of which is to transfer heat from a colder body to a hotter one.
	\item[Kelvin-Planck:] There exists no process the sole result of which is to transform heat entirely into work.
\end{description}
It is immediately clear that these formulations are more complicated this might obscure rather than reveal it's fundamental importance. Furthermore they contain notions such as `heat' and `work' which require further definition. Carath\'eodory recognized early on the central role that processes between states play a fact he used to give his own formulation \cite{Carathodory1909}.
\begin{description}
	\item[Carath\'eodory:] In every neighborhood $N_x$ of any state $x$ there exist states arbitrarily close to it that are adiabatically inaccessible.
\end{description}
This is closer in spirit to the second law since emphasis is now given on states and processes rather than heat, nonetheless it still contains terms such as `neighborhood', `arbitrarily close' and `adiabatic' which need to be precisely defined. 

In 1964 the same year Bell published his paper, R. Giles provided a rigorous mathematical proof of the second law \cite{Giles1964} based on an axiomatic formulation of thermodynamics. Though many similar treatments on the subject exist (most notably those by Lieb and Yngvasson \cite{Lieb1999},\cite{Lieb2002},\cite{Lieb2004}) it is mostly a matter of preference which one chooses to work with. 

In this section we shall present a concise formulation of thermodynamics based on the former approach, we shall then talk about adiabatic transitions and how to construct an entropy function on the set of thermodynamic states. Proofs and any theorems mentioned below as well as a more detailed discussion the interested reader can find in \cite{Giles1964}.

\subsection{Thermodynamics in ideal form.}\label{sec2.1}
The \emph{classical theory of thermodynamics} can be formulated in \emph{ideal form}, i.e. as a set of \emph{primitive terms} which constitute the building blocks of the theory, \emph{rules of interpretation} needed to give a precise definition of these terms using notions derived by common experience and a set of \emph{Axioms}. A \emph{primitive observer} capable of performing experiments and equipped with the ideal form of the theory is able to reproduce the whole structure of classical thermodynamics, using primitive terms he can construct more complex terms and using Axioms he can start proving theorems. For a more detailed discussion on the ideal form of thermodynamics see \textsection 1.1 of \cite{Giles1964}.

The first primitive term to be introduced is that of a \emph{state}. Instead of a point in phase space a \emph{state} is defined by it's method of preparation, for example a state consisting of one mole of $NaCl$ is defined by the following procedure `'mix one mole of $Na^+$ and one mole of $Cl^-$ in a test tube''. Two states $a$ and $b$ are said to be \emph{equivalent} (denoted by $a=b$) if the results of any experiment performed on the former are indistinguishable from the same experiments performed on the latter. Notice that the notion of a system (the test tube in the previous example) is included in the method of preparation and does not need to be defined separately. The set of all states will henceforth be denoted by $\mathcal S$. Following Giles's convention states will be denoted by lowercase Latin letters $a,b,c\ldots$ while systems by uppercase $A,B,C$ etc. 

The remaining two terms are relations between states, namely \emph{state addition} (denoted by $+$) and \emph{natural transition} between states (denoted by $\to$). Specifically given states $a,b\in \mathcal S$ to obtain state $c=a+b$ one brings two systems $A$ and $B$ together and simply performs simultaneously the methods of preparation for states $a$ and $b$ respectively making sure meanwhile that the systems do not interact with each other, the composite system $C$ is now in state $c$. It is interesting to note that although any state of the form $a+b$ is a state of the combined system $A+B$ the converse is not necessarily true if the two subsystems interact with each other\footnote{This is also characteristic of quantum systems where a system can be in a separable or entangled state as we shall see later.}. Multiplication of state $a$ by a positive integer $n$ results in state $na$ which consists of $n$ systems all prepared in the same state.

Suppose now that after a lapse of time $\tau$ state $a$ becomes equivalent to state $b$, we then write $a\to b$ (read as $b$ is naturally accessible from $a$), a more exact definition is given with the help of an auxiliary state
\begin{defin}\label{def2.1}
	$a\to b$ if and only if there exists state $k\in \mathcal S$ such that ${a+k\to b+k}$
\end{defin}
If for some reason the auxiliary state $k$ is practically impossible to prepare or the time $\tau$ required for the evolution is to long we can for all practical purposes write $a\not\to b$. A pair of states $(a,b)$ is called a \emph{natural process}, if $a\to b$ or $b\to a$ then the process is called \emph{possible} otherwise if $a\not\to b$ and $b\not\to a$ both hold the process is called \emph{impossible}. A natural process in which both $a\to b$ and $b\to a$ hold is called \emph{reversible}.

The ideal formulation of the classical theory of thermodynamics is completed by giving the following axioms.	
\addtocounter{axm}{1}
\begin{subaxm}\label{ax1a}
	$a+b=b+a$ 
\end{subaxm}
\begin{subaxm}\label{ax1b}
	$a+(b+c)=(a+b)+c$ 
\end{subaxm}
\begin{axm}\label{ax2}
	$a\to a$
\end{axm}
\begin{axm}\label{ax3}
	$a\to b\;\;\&\;\; b\to c \implies a\to c$
\end{axm}
\begin{axm}\label{ax4}
	$a\to b \iff a+c\to b+c$
\end{axm}
\addtocounter{axm}{1}
\setcounter{subaxm}{0}
\begin{subaxm}\label{ax5a}
	$a\to b \;\;\&\;\; a\to c \implies b\to c\quad$ or $\quad c\to b $
\end{subaxm}
\begin{subaxm}\label{ax5b}
	$b\to a \;\;\&\;\; c\to a \implies b\to c\quad$ or $\quad c\to b$
\end{subaxm}
\begin{axm}\label{ax6}
	There exists an internal state e.
\end{axm}
\begin{axm}\label{ax7}
	If there exist states $x$ and $y$ such that $na+x\to nb+y$ holds for arbitrarily large positive integers $n$ then $a\to b$.
\end{axm}
\cref{ax1a,ax1b} are simple statements of the fact that addition of states is a commutative and associative operation respectively,  \cref{ax2,ax3} express the reflexive and transitive properties of $\to$ while \cref{ax4} is a direct consequence of \cref{def2.1}. \cref{ax5a,ax5b} state that if two processes have the same initial or final states then there always exists a natural process connecting the final and initial states respectively. With the help of these Axioms it is easy to prove that processes can also be multiplied by integers.
\begin{cor}\label{cor2.1}
	If $a\to b$ then $na\to nb$ for any positive integer $n$.	
\end{cor} 
\begin{proof}
	The proof follows by induction, for $n=1$ the corollary is trivial, suppose now that it also holds true for some $n>0$, then by \cref{ax4} {$(n+1)a=na+a\to nb+a$} and $nb+a\to nb+b=(n+1)b$ , by \cref{ax3} $(n+1)a\to (n+1)b$ and since this is true for arbitrary $n$ the proof is now complete.
\end{proof}
\cref{ax6} is necessary in order to compare states, a state is called \emph{internal} if for any $a\in \mathcal S$ there exist positive integers $n,m$ and a state $c\in \mathcal{S}$ such that either $na+c\to nme$ or $nme\to na+c$ holds. The internal state serves the role of a unit of measure. \cref{ax7} is actually a consequence of a more complicated Axiom, nonetheless because of it's importance in later discussions and because a violation of this Axiom would imply that the one it is derived from is false we include it in the list. This Axiom states that a process remains unaffected if we remove at most an infinitesimal portion of the system.

\subsection{The second law of thermodynamics.}\label{sec2.2}
A question now arises, how can we distinguish a possible from an impossible process? This is easily accomplished by introducing a set of functions called the \emph{components of content} and with the help of the following theorem.
\begin{thrm}\label{thrm2.1}
	There exists a sufficient number of real valued functions of state $Q:\mathcal{S}\to\mathbb{R}$ called components of content with the following properties:
	\begin{enumerate}[i)]
		\item $Q(a+b)=Q(a)+Q(b)$
		\item\label{thrm1b} $(a,b)$ is a possible process if and only if $Q(a)=Q(b)$ for every component of content $Q$.
	\end{enumerate}
\end{thrm}
These components can be simply interpreted as the constraints during which a process takes place, typical examples are the volume of a system, its mass and it's total energy. The set of components of content forms a vector space called \emph{content space}. 

The second law of thermodynamics can also be stated as a theorem
\begin{thrm}[Second law of thermodynamics under natural processes.]\label{thrm2.2}
	There exists a positive real valued function $S$ on the set of states $S:\mathcal{S}\to\mathbb{R}$  called a quasi-entropy function with the following properties.
	\begin{enumerate}[i)]
		\item $S(a+b)=S(a)+S(b)$
		\item $a\to b\;\;\&\;\;b\not\to a \implies S(a)< S(b)$
		\item $a\to b\;\;\&\;\;b\to a \implies S(a)=S(b)$ 
	\end{enumerate}
\end{thrm}
The proofs of \cref{thrm2.1,thrm2.2} can be found in \textsection A.3 and \textsection A.4 of \cite{Giles1964} respectively. It is important to note that \cref{ax6,ax7} are essential in the construction of these proofs, their validity is thus crucial for the theory.

It is not hard to see that a quasi-entropy function is unique up to a positive scale factor and a change of reference.
\begin{lem}\label{lem2.1}
	Given a quasi-entropy function $S$ if $c$ is a positive constant and $Q$ a component of content then the function given by $S'=cS+Q$ is also a quasi-entropy function.
\end{lem}
\begin{proof}
	$S'$ satisfies the properties of \cref{thrm2.1}
	\begin{enumerate}[i)]
		\item $S'(a+b)=cS(a+b)+Q(a+b)=cS(a)+Q(a)+cS(b)+Q(b)=S'(a)+S'(b)$
		\item $S(a)\leq S(b)\iff \frac{S'(a)-Q(a)}{c}\leq \frac{S'(b)-Q(b)}{c}$ and since $(a,b)$ is a possible process it follows that $S'(a)\leq S'(b)$.
		\item Interchanging $a$ and $b$, $b\to a$ implies now that $S'(b)\leq S'(a)$ and since also $S'(a)\leq S'(b)$, $S'(a)=S'(b)$ follows.\qedhere
	\end{enumerate}
\end{proof}
\cref{lem2.1} implies that changing the unit in which quasi-entropy is being measured has no effect on the way the states are ordered.

\subsection{Adiabatic processes.}\label{sec2.3}
Up to this point there has been no mention of mechanical states or adiabatic processes, the theory of classical thermodynamics can be formulated without the need to introduce either one. Nonetheless the second law assumes a more useful formulation when the theory is extended to include such notions. This is because the relation induced by the natural transition between states turns $\mathcal S$ into a partially ordered set, i.e. there exist states $a,b\in \mathcal S$ such that neither $a\to b$ nor $b\to a$ hold. By introducing adiabatic processes the set becomes totally ordered parametrized uniquely by the entropy. To achieve this one more Axiom needs to be added to the list.
\begin{axm}\label{ax8}
	For any state $a$ there exists an anti-equilibrium state $x$ such that $x\to a$. If $x$ and $y$ are two such states then so is $x+y$.
\end{axm}
An \emph{anti-equilibrium state} is defined as follows.
\begin{defin}\label{def2.2}
A state $x\in\mathcal{S}$ is called an anti-equilibrium state if there exists no state $a\in\mathcal{S}$ such that $a\to x$ and $x\not\to a$.
\end{defin}
The set of \emph{mechanical states} denoted by $\mathcal M$ is simply a subset of the set of anti-equilibrium states. Since this subset can be arbitrarily defined $\mathcal{M}$ is chosen for simplicity to be the set of states satisfying \cref{ax8}. Note that by definition any process with initial and final mechanical states is reversible, also if $k$ is a mechanical state then so is $nk$ for any positive integer $n$. With the help of these states a definition of \emph{adiabatic} processes can now be given.
\begin{defin}\label{def2.3}
	State $b$ is said to be \emph{adiabatically accessible} by $a$ (written $a\prec b$) if there exist \emph{mechanical states} $k$ and $l$ such that $a+k\to b+l$.
\end{defin}
Any possible natural process is also adiabatic, if $a\to b$ then adding any state $k\in\mathcal{M}$ on both sides implies that $a\prec b$, the converse is not necessarily true. 

It is an easy task to prove that all of the Axioms introduced so far remain essentially unchanged if $\to$ is replaced by $\prec$ the same of course applies to any Theorem, Lemma, Corollary and Definition. \crefrange{ax2}{ax8} can then be substituted by:
\addtocounter{paxm}{1}
\begin{paxm}\label{pax2}
	$a\prec a$
\end{paxm}
\begin{proof}
As was already mentioned any possible natural process is also adiabatic.
\end{proof}
\begin{paxm}\label{pax3}
	$a\prec b\;\;\&\;\; b\prec c \implies a\prec c$
\end{paxm}
\begin{proof}
	Suppose $k,l,h,j\in \mathcal{M}$ such that $a+k\to b+l$ and $b+h\to c+j$, by \cref{ax4} $a+k+h\to b+l+h$ and $b+h+l\to c+j+l$ applying \cref{ax3} $a+k+h\to c+j+l$ and since according to \cref{ax8} the sum of two mechanical states is a mechanical state $a\prec c$ follows.
\end{proof}
\begin{paxm}\label{pax4}
	$a\prec b \iff a+c\prec b+c$
\end{paxm}
\begin{proof}
	Suppose there exist states $k,l\in \mathcal{M}$ such that $a+k\to b+l$ by \cref{ax4} ${a+c+k\to b+c+l}$ so $a+c\prec b+c$. Conversely $a+c\prec b+c$ implies that there exists $k,l\in\mathcal{M}$ such that $a+c+k\to b+c+l$, by \cref{ax4} $a+k\to b+l$ and $a\prec b$.
\end{proof}
\addtocounter{paxm}{1}
\begin{paxm}\label{pax6}
	There exists an internal state e.
\end{paxm}
\begin{proof}
Same as that of \cref{pax2}.
\end{proof}
\begin{paxm}\label{pax7}
	If there exist states $x$ and $y$ such that $na+x\prec nb+y$ holds for arbitrarily large  positive integers $n$ then $a\prec b$.
\end{paxm}
\begin{proof}
Suppose there exists a sequence of mechanical states $k_n,l_n$ and ${x,y\in\mathcal{S}}$ such that for every positive integer $n$ $na+x+k_n\to nb+y+l_n$. By definition any mechanical state is accessible by any other under a natural process, this means that $nb+y+l_n\to nb+y+k_n$, combining this with the previous process results in $na+x+k_n\to nb+y+k_n$ and by \cref{ax4} $na+x\to nb+y$. Finally  by \cref{ax7} $a\to b$ so $a\prec b$. 
\end{proof}
\begin{paxm}\label{pax8}
	For any state $a$ there exists an anti-equilibrium state $x$ such that $x\prec a$. If $x$ and $y$ are two such states then so is $x+y$.
\end{paxm}
\begin{proof}
	Since $x\to a$ for any $k\in\mathcal{M}$, $x+k\to a+k$ so $x\prec a$.
\end{proof}
\cref{ax5a,ax5b} are no longer necessary since under adiabatic processes $\mathcal S$ becomes a totally ordered set.
\begin{thrm}\label{thrm2.3}
	Given any two states $a,b\in \mathcal S$ either
	$$a\prec b\quad\mbox{ or }\quad b\prec a$$
\end{thrm}
\begin{proof}
	Suppose $k,l\in \mathcal M$ such that $k\to a $ and $l\to b$. By \cref{ax4} ${k+l\to a+l}$ and $k+l\to k+b$ and by \cref{ax5a} either $a+l\to k+b$ or $k+b\to a+l$. It immediately follows that either $a\prec b$ or $b\prec a$.
\end{proof}
According to \cref{thrm2.3} every adiabatic process $(a,b)$ is now possible, components of content are therefore redundant in this case and are no longer needed.

\cref{thrm2.2} can now be replaced by.
\addtocounter{pthrm}{1}
\begin{pthrm}[Second law of thermodynamics under adiabatic processes]\label{pthrm2.2}
	There exists a positive real valued function $S$ on the set of states $S:\mathcal{S}\to\mathbb{R}$ called an entropy function with the following properties.
	\begin{enumerate}[i)]
		\item $S(a+b)=S(a)+S(b)$
		\item $S(k)=0\quad\forall k \in \mathcal M$ 	
		\item $a\prec b\;\;\&\;\;b\not\prec a \iff S(a)< S(b)$
		\item $a\prec b\;\;\&\;\;b\prec a \iff S(a)=S(b)$
	\end{enumerate}
\end{pthrm}
\begin{proof}
$S$ is the same function as the one in \cref{thrm2.2} so additivity follows immediately. As was mentioned earlier for any $k\in\mathcal{M}$ $(k,nk)$ is a reversible process for any positive integer $n$, this implies that $S(k)=nS(k)$ from which $S(k)=0$. The rest of the properties follow immediately from \cref{thrm2.2,thrm2.3} and properties i) and ii).
\end{proof}
At this point we must draw attention to the difference between Theorems \ref{thrm2.2} and \ref{pthrm2.2}. Specifically under natural transitions the condition that state $a$ has less entropy than state $b$ is only a necessary condition for $a\to b$ while under adiabatic transitions the same condition is also sufficient  this is of course a direct consequence of \cref{thrm2.3}. For this reason we shall refer to \cref{thrm2.2} as the \emph{weak form of the second law} and \cref{pthrm2.2} as the \emph{strong form}. The uniqueness of an entropy function has one more constraint than the one for quasi-entropy.
\begin{lem}\label{lem2.2}
	Given an entropy function $S$ if $c$ is a positive constant and $Q$ a non-mechanical component of content then the function given by ${S'=cS+Q}$ is also an entropy function. A component of content $Q$ is called \emph{non-mechanical} if $Q(k)=0,\forall k \in \mathcal{M}$.
\end{lem}
\begin{proof}
	$S'$ satisfies the properties of \cref{pthrm2.2}.
	\begin{enumerate}[i)]
		\item $S'(a+b)=bS(a+b)+Q(a+b)=cS(a)+Q(a)+cS(b)+Q(b)=S'(a)+S'(b)$
		\item If $k\in \mathcal M$ then $S'(k)=cS(k)+Q(k)=Q(k)$ which equals zero because $Q$ is a non-mechanical component of content.
		\item $S(a)\leq S(b)\iff \frac{S'(a)-Q(a)}{c}\leq \frac{S'(b)-Q(b)}{c}$ and since $(a,b)$ is a possible process $Q(a)=Q(b)$ and $S'(a)\leq S'(b)$.
		\item Interchanging $a$ and $b$ $b\to a$ now implies $S'(b)\leq S'(a)$ and since also $S'(a)\leq S'(b)$, $S'(a)=S'(b)$ follows easily.\qedhere	
	\end{enumerate}
\end{proof}
The amount of entropy a state has can now be determined quantitatively. The first step is to choose a reference state $e$ and use it's entropy as the unit in which entropy is measured for all other states. Since $\mathcal S$ is totally ordered under adiabatic processes one simply constructs for any state $a\in\mathcal{S}$ the following sets:
\begin{equation}\label{eq2.1}
U(a):=\Set{\frac{m}{n}|na\prec me}
\end{equation}
and
\begin{equation}\label{eq2.2}
L(a):=\Set{\frac{m}{n}|me\prec na}
\end{equation}
It is not hard to see that any element of $U(a)$ is an upper bound to any element of $L(a)$ and conversely any element of $L(a)$ a lower bound to $U(a)$ in fact it can be shown that.
\begin{thrm}\label{thrm2.4}
	$\max{L(a)}=\min{S(a)}$.
\end{thrm}
\begin{proof}
	Suppose that $\frac{m'}{n'}$ and $\frac{m}{n}$ are the maximum and minimum of $L(a)$ and $U(a)$ respectively, since $\mathcal{S}$ is totally ordered under adiabatic transitions either $na+n'a\prec me+m'e$ or $me+m'e\prec na+n'a$ must hold.
	\begin{enumerate}[a)]
		\item If $na+n'a\prec me+m'e$ then $\frac{m}{n}\leq\frac{m+m'}{n+n'}\implies\frac{m}{n}\leq \frac{m'}{n'}$
		\item If $me+m'e\prec na+n'a$ then $\frac{m+m'}{n+n'}\leq\frac{m'}{n'}\implies\frac{m}{n}\leq \frac{m'}{n'}$
	\end{enumerate} 
but $\frac{m'}{n'}\leq\frac{m}{n}$ so $\frac{m'}{n'}= \frac{m}{n}$ follows. \qedhere
\end{proof}
\cref{thrm2.4} implies the existence of a unique reversible process between multiple copies of $a$ and $e$.
\begin{thrm}\label{thrm7}
	There exist positive integers $n$ and $m$ such that $$na\sim me$$ 
\end{thrm}
\noindent where $\sim$ is used to denote that both $na\prec me$ and $me\prec na$ hold.
\begin{proof}
	Suppose once more that $\frac{m'}{n'}$ and $\frac{m}{n}$ are the maximal and minimal elements of $L(a)$ and $U(a)$ respectively so that $na\prec me$ and $m'e\prec n'a$, by \cref{thrm2.4} $n'na\prec n'me=nm'e$ from which $n'a\prec m'e$, $me\prec na$ is proven similarly.
\end{proof}
With the help of \cref{thrm2.4} the entropy of $a$ is defined as the Dedekind cut of  $U(a)$ and $L(a)$ (see \cref{fig2.1}).
\begin{figure}
	\begin{center}
		\begin{tikzpicture}
			\fill (0,0)circle(1pt)node[left]{$0$};
			\fill[gray,opacity=0.5] (3,-0.5)rectangle(10,0.5);
			\draw(1.5,-0.5)node[below]{$L(a)$};
			\draw(6.5,-0.5)node[below]{$U(a)$};
			\draw[thick,->](0,0)--(10,0)node[right]{$\mathbb{R}$};
			\draw[dashed](3,-0.8)node[below]{$S(a)$}--(3,0.8);
		\end{tikzpicture}
	\end{center}
\caption{Entropy as the Dedekind cut of sets $U(a)$ and $L(a)$.}\label{fig2.1}
\end{figure}
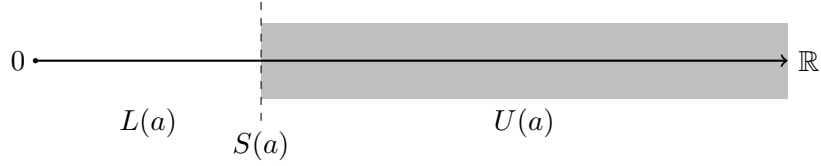
\begin{equation}\label{eq2.3}
	S(a):=\sup L(a)=\inf U(a)
\end{equation}
The total ordering of the set of states under adiabatic transitions is thus equivalent to the existence of a unique function of state which in turn uniquely characterizes the ordering through the strong form of the second law.

Any theory where notions of state, state unions and transitions between states can be defined is formally equivalent to the classical theory of thermodynamics as long as it satisfies Axioms \ref{ax1a}-\ref{ax8}. If this is the case then an entropy function automatically exists obeying the weak form of the second law. By extending the theory to include adiabatic transitions the set of states will now become totally ordered and it will be possible to quantitatively determine the entropy of any state by identifying a reversible process between $n$ copies of the state and $m$ copies of a suitably chosen reference state, this entropy function will now satisfy the strong form of the second law and can be employed to determine which states are adiabatically accessible by others in a unique way. A possible candidate for such a theory is the theory of \emph{entanglement}. Though at present it is not known if in general it possesses a thermodynamical structure, for some simple special cases the two theories become formally equivalent.
\section{Entanglement.}\label{chap3}
In order to compare thermodynamics with the theory of entanglement, a formal definition of the later must be given first. This is the objective of the current section. We will only be interested in discrete bipartite quantum systems of a finite dimensional Hilbert space each composite system of which is held by a different party. Despite it's simplicity this case is far from trivial as we don't possess a complete picture of the theory and many questions remain unanswered to this day. A very good review which covers a wide range of subjects on the theory with an exhaustive list of references can be found in \cite{Horodecki2009}.

\paragraph{\textmd {Notations:}} A general bipartite system will be referred to as a $d_A\times d_B$ system where $d_X=dim\mathcal{H}_X$ and $\hilb{A}$ and $\hilb{B}$ are the Hilbert spaces of subsystems $A$ and $B$ respectively. The set of density matrices acting on Hilbert space $\hilb{}$ i.e. positive semi-definite operators with unit trace will be denoted by $\mathcal D_{\mathcal{H}}$. $\norm{\rho}=Tr\sqrt{\rho^\dagger\rho}$ is the trace norm in Hilbert-Schmidt space. Ket notation for pure states will usually be employed only for eigenbases and the computational basis. Any logarithms appearing in the text are of base two.

\subsection{Separable and Entangled states.}\label{sub3.1}
A bipartite entangled system consists of a pair of particles prepared at the same point in spacetime initially interacting with each other under the influence of a common interaction Hamiltonian $H_{int}$. Each particle is subsequently sent to a different lab separated by a spacelike distance where physicists Alice and Bob intercept them.  The pair of particles is now in an \emph{entangled state} which possesses ``spooky action at a distance'' as Einstein called it i.e. instantaneous non-local correlations between the results of suitably chosen experiments that the two observers perform on their respective particles. 

The simplest case to consider is when the system is in a pure state. The first formal definition of entanglement was given by Werner \cite{Werner1989}.
\begin{defin}\label{def3.1}
A bipartite state $\Phi_{\scriptscriptstyle AB}$ of a $d_A\times d_B$  system is called \emph{separable} (\emph{entangled}) if it can (cannot) be written as a product state	
	\begin{labeq}{eq3.1}
		\Phi_{\scriptscriptstyle AB}=\varphi_{\scriptscriptstyle A}\otimes\psi_{\scriptscriptstyle B}
	\end{labeq}
with $\varphi_{\scriptscriptstyle A}\in\mathcal{H}_A$ and $\psi_{\scriptscriptstyle B}\in\mathcal{H}_B$.
\end{defin}
The classic example of a pure entangled state is given by the singlet state of a $2\times2$ (qubit-qubit) system given in \cref{eq1.1}. This state along with the following
\begin{align}
	\Psi^+_{\scriptscriptstyle AB}&=\frac{1}{\sqrt{2}}(\Ket{01}+\Ket{10}) \\
	\Phi^+_{\scriptscriptstyle AB}&=\frac{1}{\sqrt{2}}(\Ket{00}+\Ket{11}) \\
	\Phi^-_{\scriptscriptstyle AB}&=\frac{1}{\sqrt{2}}(\Ket{00}-\Ket{11})\label{eq3.3}
\end{align}
constitute an orthonormal basis of the four dimensional Hilbert space of the composite system known as a \emph{Bell basis}. 

For any state $\varphi$ it is easy to infer whether it is entangled or not by simply writing down its \emph{Schmidt decomposition}
\begin{labeq}{eq3.4}
	\varphi_{\scriptscriptstyle AB}=\sum_{i=1}^{d}\sqrt{\lambda_i}\ket{ii}
\end{labeq}
where $d=\min\set{d_A,d_B}$ is the \emph{Schmidt rank} of the state and $\lambda_i$ are real positive numbers with $\sum_i\lambda_i=1$  known as the \emph{Schmidt coefficients} (for proof see \cref{appa}). If the coefficients are ordered in decreasing order such that $\lambda_1\geq\lambda_2\geq\ldots\geq\lambda_d$ then \cref{eq3.4} is called an \emph{ordered Schmidt decomposition} (OSC). It is immediately clear that the state is separable if d=1 and entangled otherwise. 

By taking the partial trace of \cref{eq3.4} it is easy to show that each observers subsystem is in a mixed state.

\begin{labeq}{eq3.5}
\rho_{\scriptscriptstyle X}=\sum_i\lambda_i\ketbra{i}{}
\end{labeq}
where $X=A,B$. 

A \emph{maximally entangled} state in $d$ dimensions is given by
\begin{labeq}{eq3.6}
	\Psi_d^+=\frac{1}{\sqrt{d}}\sum_{i=1}^{d}\indexket{ii}{\scriptscriptstyle AB}
\end{labeq}
the reason why the term maximally is used for this state will become clearer later. According to this definition all of the states in the Bell basis are maximally entangled states of the 4-dimensional Hilbert space.

With the help of the pure state case it is now easy to extend \cref{def3.1} to include systems in mixed states. 
\begin{defin}\label{def3.2}
A mixed state is called \emph{separable} (\emph{entangled}) if it can (cannot) be written as a convex combination  of separable pure states.
	\begin{labeq}{eq3.7}
		\rho_{\scriptscriptstyle AB}=\sum_i p_i\ketbra{\varphi_i}{}\otimes\ketbra{\psi_i}{}
	\end{labeq}
with $p_i\geq 0$ and $\sum_i p_i=1$.
\end{defin}
Note that in the special case where $p_i=\delta_{ij}$ \cref{def3.1,def3.2} coincide. In contrast to the pure case it is usually very hard to determine if a given mixed state is entangled or not. One reason for this difficulty lies in the fact that a mixed state can be written in infinitely many ways, another is that by definition a mixed state already possesses a degree of correlation given by the probability distribution in \cref{eq3.7}. This complicates matters when one is trying to distinguish which correlations present in a mixed state are of classical origin or are purely quantum in nature. It is imperative then to try and devise new methods to determine if a given mixed state is entangled or not.

For $2\times 2$ and $2\times 3$ systems there exists a necessary and sufficient condition for separability discovered by Peres \cite{Peres1996} and proven as a theorem in \cite{Horodecki1996}.
\begin{thrm}\label{thrm3.1}
A bipartite mixed state $\rho$ is separable if and only if 
	\begin{labeq}{eq3.8}
		\rho^{T_X}\geq 0\qquad X=A,B
	\end{labeq}
\end{thrm}
Recall that a density matrix is positive if for any state $\ket{\phi}$ of the system's Hilbert space $\braket{\phi|\rho|\phi}\geq 0$. $T_X$ denotes the partial transposition with respect to the chosen observer's basis, for example if $X=A$ then the components of $\rho^{T_A}$ are given by
\begin{labeq}{eq3.9}
	(\rho_{\mu m,\nu n})^{T_A}=\rho_{\nu m,\mu n}
\end{labeq} 
where Greek indices refer to Alice's components and Latin to Bob's. It is very interesting to note that for such systems there is a direct connection between separability and the local time arrow as was pointed out in \cite{Sanpera1997}. It can be shown that partial transposition is directly related to time reversion in one of the subsystems, this affects only entangled states causing their partially transposed matrices to have negative eigenvalues.

For systems of higher dimensions the condition of positivity of the partial transpose is a necessary but not a sufficient condition. Reference \cite{Bruss2001} contains a short review on the separability problem.
\subsection{Local Operations and Classical Communication.}\label{sub3.2}
Most of the applications mentioned in the introduction require maximally entangled states to be shared between two parties, in order to be implemented without errors. These states are very difficult to maintain for long periods of time due to interaction with their environment, for this reason they are considered to be a valuable resource in the context of quantum information theory. On the other hand if Alice and Bob are allowed to communicate classically it is easy for them to prepare any separable state they wish even if it is mixed\footnote{Alice simply determines the value of a discrete random variable $i$ occurring with probability $p_i$, she then informs Bob about her outcome via a classical channel (e.g a phone line) after which they both prepare the state $\ketbra{\varphi_i}{A}$ and $\ketbra{\psi_i}{B}$. Eventually they discard the information of the original outcome by placing their particles in a reservoir. After repeating the procedure a large number of times, if they now draw at random from the reservoir a pair of entangled particles their combined state will be given by \cref{eq3.7}.}, these states are therefore considered as a free resource. Suppose now that being limited by a budget they share as many entangled states as possible and although they cannot create any new ones by sending particles to each other are nonetheless allowed to exchange as many classical messages they wish, the question that naturally arises is what kinds of operations can they both implement under the above restriction. The answer is the set of \emph{Local operations and Classical Communication} or (LOCC) for short. This set consists of local quantum operations carried out on each observers subsystem $\rho_A=\tr{B}{\rho}$ and $\rho_B=\tr{A}{\rho}$ assisted by an unlimited amount of classical communication between their laboratories, this allows them to correlate their actions enhancing their operational capabilities. 

Mathematically a quantum operation is a linear function on the set of density matrices $\mathcal{D}_\mathcal{H}$ of a given Hilbert space $\mathcal{H}$. If $\Lambda$ is such an operation with the property that any state $\rho\in\mathcal{D}_\mathcal{H}$ is mapped to a positive operator of a different Hilbert space $\mathcal{H'}$ then it is called a \emph{positive quantum operation}. The effect of any quantum operation on a state can be written with the help of the \emph{operator-sum decomposition}
\begin{labeq}{eq3.10}
	\Lambda(\rho)=\sum_iE_i\rho E^{\dagger}_i
\end{labeq}
where $E_i$ are the so called \emph{Kraus operators}. If $\sum_iE^{\dagger}_iE_i=I$ it is easy to show that $\tr{}{\Lambda(\rho)}=1$, in this case $\Lambda$ is said to be trace preserving and is called a \emph{deterministic} quantum operation. A deterministic operation always maps density matrices to density matrices. For a non-trace preserving operation $\sum_iE^{\dagger}_iE_i\leq I$ and $\tr{}{\Lambda(\rho)}\leq1$. 

A \emph{generalized measurement} carried out on state $\rho$ is described by an ensemble $\Set{p_i,\rho_i}$ where state $\rho_i=\frac{\mathcal{E}_i(\rho)}{\tr{}{\mathcal{E}_i(\rho)}}$ is measured with probability $p_i=\tr{}{\mathcal{E}_i(\rho)}$ and $\mathcal{E}_i$ are non-trace preserving operations. Any deterministic operation can be written as the ensemble average of a generalized measurement such that
\begin{labeq}{eq3.11} 
\Lambda(\rho)=\sum_i\mathcal{E}_i(\rho)=\sum_ip_i\rho_i
\end{labeq}
The most simple examples of deterministic operations and generalized measurements are the unitary operation and von-Neumann measurements respectively. For a brief introduction to the theory of quantum operations see \cite{Schumacher1996,Nielsen1997}.

A LOCC operation can now be described by a sequence of successive applications of an Alice quantum operation of the form $\Lambda_A\otimes I_B$ and a Bob operation $I_A\otimes\Lambda_B$ assisted by classical communication after each step. Depending on whether communication is made in only one direction or two LOCC operations are divided into two classes 1-LOCC and 2-LOCC respectively.

The characteristic property of such a transformation is that it is impossible to create entanglement from any separable state by using LOCC operations alone, this property which can be taken as an alternative definition of LOCC imposes some strong restrictions on the theory.
\begin{lem}\label{lem3.1}
If $\Lambda\in LOCC$ then
	\begin{labeq}{eq3.13}
		\forall \sigma\in \mathcal S \qquad \Lambda(\sigma)\in \mathcal S
	\end{labeq}
where $\mathcal S$ denotes the set of separable states.
\end{lem}
\begin{proof}
	By definition LOCC is a sequence of either Alice or Bob local quantum operations, it is easy then to see that any such operation acting on a separable state will always result in a new separable state.\qedhere
\end{proof}

By lifting certain restrictions it is possible to define new classes of operations which might allow processes that were previously impossible under LOCC. For example the two parties might be allowed to exchange quantum information as well. This will prove useful in deriving the second law for entanglement processes in later sections. An interesting example of such a class of operations is the set of \emph{separable operations} with Kraus operators of the form $E_i=A_i\otimes B_i$, although the set of one or two way LOCC operations belongs to this set the converse is not true \cite{Rains1997,Bennett1999b}.

\subsection{Distillation and dilution.}\label{sec3.3}

As was already mentioned maximally entangled states are extremely sensitive to noise making them very hard to transfer and store. In practice due to interaction with their environment they quickly decohere into a less entangled mixed state.
\begin{labeq}{eq3.15}
	\Psi_d^+\to\rho
\end{labeq}
This poses a major problem in quantum information theory since mixed states perform poorly if used in various protocols. This has led to the development of methods with the aim of reversing \cref{eq3.15} allowing the recovery of maximally entangled states from less entangled ones known as \emph{entanglement distillation protocols}. 

The first of these types of protocols was introduced by Bennett et al. for the special case where two parties share a large number of copies of a given pure state \cite{Bennett1996} . Without loss of generality suppose that Alice and Bob wish to distill a pure state $\phi$ of a $2\times 2$ system. In \cref{appa} it is shown that this state can always be written in it's Schmidt decomposed form
\begin{labeq}{eq3.16}
	\phi_{\scriptscriptstyle AB}=\alpha\ket{00}+\beta\ket{11}
\end{labeq}
where $\alpha$ and $\beta$ are real constants with $\alpha^2+\beta^2=1$. With the help of an operator $S$ which acts on product states defined by 
\begin{labeq}{eq3.17}
	S(\phi\otimes\psi)=\phi\otimes\psi+\psi\otimes\phi
\end{labeq}
the state composed of $n$ copies of \cref{eq3.16} can now be written as
\begin{labeq}{eq3.18}
	\phi^{\otimes n}_{\scriptscriptstyle AB}=\sum_{k=0}^{n}\alpha^k\beta^{n-k}S\left(\multiket{00}{k}\multiket{11}{n-k}\right)
\end{labeq}
The set $\left\{\ket{S_k}\right\}_{k=1}^n$ given by $\ket{S_k}=\binom{n}{k}^{-\frac{1}{2}}S\left(\multiket{00}{k}\multiket{11}{n-k}\right)$ forms an orthonormal set of a $2^{2n}$ dimensional Hilbert space. Each state in this set is a normalized sum of all $\binom{n}{k}$ possible sequences containing $k$ copies of $\ket{00}$ and $n-k$ copies of $\ket{11}$ with equal weights. 

Alice begins the protocol by performing a von-Neumann measurement to determine the value of $k$. If $\Pi_k^n$ denotes the set of all possible permutations of $k$ zeros and $n-k$ ones, then $\{\ketbra{s}{A}\}_{s\in \Pi_k^n}$ is a set of orthonormal operators in $\hilb{A}$ and the probability of obtaining result $k$ is given by 
\begin{labeq}{eq3.19}
	p_k=\sum_{s\in \Pi_k^n}\braket{s|\tr{B}{\ketbra{\phi^{\otimes n}}{}}|s}=\binom{n}{k}a^{2k}b^{2(n-k)}
\end{labeq}
note that $p_k$ is simply the square of the coefficient appearing in front of $\ket{S_k}$ in \cref{eq3.18}. After she has completed her measurement Alice then informs Bob of the result\footnote{In this case classical communication is not necessary since Bob can always perform the same measurement on his subsystem which will always result in obtaining result $k$.}.

In the limit of a large number of initial copies it is an immediate consequence of the law of large numbers that there exist $\epsilon_n>0$ and $\delta_n>0$ with $\lim_{n\to\infty}\epsilon_n\to 0$ and  $\lim_{n\to\infty}\delta_n\to 0$ such that  
\begin{labeq}{eq3.20}
	\sum_{k=\lfloor n(\alpha^2-c\delta_n)\rfloor}^{\lceil n(\alpha^2+c\delta_n)\rceil}p_k\geq(1-\epsilon_n)
\end{labeq} 
where $c$ is a positive constant. The left hand side of \cref{eq3.20} is simply the total probability that a given sequence belongs to the set of \emph{typical sequences}. The size of this set lies between
\begin{labeq}{eq3.21}
	(1-\epsilon_n)2^{n(H(\alpha^2)-\delta_n)}\leq\#(typ)\leq 2^{n(H(\alpha^2)+\delta_n)}
\end{labeq}
where $H(\alpha^2)=-\alpha^2\log{\alpha^2}-(1-\alpha^2)\log{(1-\alpha^2)}$ is the binary Shannon entropy. Since the set of typical sequences is also orthonormal, $\#(typ)$ is the dimension of the subspace of the original Hilbert space spanned by this set.

 In the limit of large $n$, it should be obvious that by applying the above protocol Alice and Bob have succeeded in obtaining $\ket{S_{k\approx n\alpha^2}}$ with a probability tending to one. Comparing this state with \cref{eq3.4} and noting that it is a normalized sum of all possible typical sequences with equal weights it immediately follows that it is a maximally entangled state of at least $d\approx2^{n(H(\alpha^2)-\delta_n)}$ dimensions which is unitarily equivalent to approximately $n(H(\alpha^2)-\delta_n)$ copies of $\Psi_2^+$, note that to transform this state into a singlet Bob simply has to rotate his subsystem by applying a $\sigma_x\sigma_z$ unitary operation. 
 
 If $\eta$ is the number of singlet states obtained by distillation divided by the number of input states then
\begin{labeq}{eq3.22}
	\eta\leq H(\alpha^2)
\end{labeq}
where equality holds asymptotically. The above protocol is easily generalized  for distillation of any $d_A\times d_B$ pure state, in this case the right hand side of \cref{eq3.22} is given be the Shannon entropy of the state's Schmidt coefficients, this is always less than $\log d$ where d is the Schmidt rank. It is therefore impossible to distill a d-dimensional maximally entangled state from a less entangled pure state. Various distillation protocols also exist for mixed states \cite{Bennett1996b},\cite{Bennett1996a}.

\emph{Entanglement dilution} is the exact opposite of what we have just described. For the case of pure states a simple protocol for diluting a maximally entangled state into a less entangled one is given by the following procedure. Initially Alice and Bob share a maximally entangled state $\Psi^+_d$. Alice then prepares locally a pair of entangled particles in the state $\phi_{\scriptscriptstyle A'A}$ on her side\footnote{She is always allowed to do so since in creating the state no entanglement is being created between her and Bob.} and teleports (see \cref{appb} for more details) one of the particles to Bob this is described by the following process
\begin{labeq}{eq3.23}
	\phi_{\scriptscriptstyle A'A}\otimes\indexket{\Psi^+_d}{AB}\to\indexket{\Psi^+_d}{A'A}\otimes\phi_{\scriptscriptstyle AB}
\end{labeq}
In the end they will end up sharing the desired state $\phi_{\scriptscriptstyle AB}$. This is not an efficient protocol since the entanglement that was needed to create $\phi_{\scriptscriptstyle AB}$ (the d-dimensional maximally entangled state) is more than the entanglement that can be distilled by it. 

To improve efficiency Alice must start with a large number of copies of state $\phi_{\scriptscriptstyle A'A}$. More precisely suppose that she has prepared locally n copies of the state given by \cref{eq3.18} where $B$ is replaced everywhere by $A'$. To teleport each $\ket{S_k}$ she would have to share with Bob a $d_k=\binom{n}{k}$ dimensional maximally entangled state, since this will happen with probability $p_k$ the expected amount of singlets they must share is equal to
\begin{labeq}{eq3.24}
\sum_kp_k\log{d_k}
\end{labeq}
but as was shown earlier in the limit of large $n$, $k\approx n\alpha^2$, $p_k\approx1$ and the maximally entangled state has at most $d_k\approx2^{n(H(\alpha^2)+\delta_n)}$  dimensions. Thus in order for them to obtain the desired state by entanglement dilution they must use approximately $n(H(\alpha^2)+\delta_n)$ singlets. If $\eta'$ is the number of singlets needed for dilution divided by the number of obtained states then
\begin{labeq}{eq3.25}
	\eta'\geq H(\alpha^2)
\end{labeq}
with equality holding in the asymptotic limit. For a general pure state the right hand side is again equal to the Shannon entropy of the state's Schmidt coefficients.

In the limit of an asymptotic number of copies it is easy to see that for any pure entangle state
\begin{labeq}{eq3.26}
	\eta=\eta'=-\sum_i\lambda_i\log\lambda_i
\end{labeq}
so distillation and dilution become mutually reversible processes, the amount of singlets invested in creating the state can always be entirely recovered.

It would seem physically reasonable to assume that every entangled state can be distilled to the singlet form, however this is false. In \cite{Horodecki1998c} it has been proven that a necessary condition for distillability is given by the negativity of the partial transpose. The basic idea behind the proof of this theorem is that if a state is to be distillable then for a certain number of copies a $2\times 2$ subspace of $\rho^{\otimes n}$ will be in a state that closely approximates a singlet, since by \cref{thrm3.1} a $2\times 2$ entangled state has a negative partial transpose then so must the original state.
\begin{cor}\label{cor3.1}
If for any mixed state $\rho\in\mathcal{D_{\mathcal{H}}}$ 
	\begin{labeq}{eq3.27}
		 \rho^{T_X}\geq 0 \qquad X=A,B
	\end{labeq}
then it cannot be distilled.
\end{cor}

By combining \cref{thrm3.1} with \cref{cor3.1} it follows that any entangled state of a $2\times 2$ and $2\times 3$ system is distillable. For systems of higher dimension there exist mixed entangled states that have a positive partial transposition \cite{Horodecki1997a} and therefore cannot be distilled, such states are called \emph{bound entangled states}.
\subsection{Entanglement measures.}\label{sub3.4}
Since entanglement is a resource it is desirable to find a way to quantify the amount of entanglement present in a given quantum state. This is possible with the help of an \emph{entanglement measure} which is a real valued function on the set of states $\mathcal{D_\mathcal{H}}$. To help understand the qualitative features of the theory, a discussion of its properties is necessary before a formal definition can be given.

Suppose $E$ is such a measure mapping the set of states of Hilbert space $\mathcal{H}$ onto the real numbers. By definition separable states contain no entanglement so their measure must equal	 zero.
\begin{prop}[Compatibility]\label{prop3.1} 
	\begin{labeq}{eq3.28}
		\forall \sigma\in \mathcal S\quad E(\sigma)=0 
	\end{labeq}
\end{prop}
When performing a LOCC operation we would like to know what the entanglement of the initial and final states are. The simplest kind of such an operation is a unitary transformation. Recalling that the effect of a unitary operator $U$ is simply to rotate the system's Hilbert space it is natural to demand that during such a rotation the amount of entanglement remains constant.
\begin{prop}\label{prop3.2} For any bilocal Unitary operation $U_A\otimes U_B$
	\begin{labeq}{eq3.29}
		E\((U_A\otimes U_B)\rho( U_A^\dagger\otimes U_B^\dagger)\)=E(\rho)
	\end{labeq}
\end{prop}
\noindent Since by definition it is impossible to create entanglement by LOCC operations it also follows that
\begin{prop}[Monotonicity under LOCC]\label{prop3.3} 
For any $\Lambda\in LOCC$
	\begin{labeq}{eq3.30}
		E(\rho)\geq E(\Lambda(\rho))
	\end{labeq}
\end{prop}
\noindent A stronger form of monotonicity is given by demanding that entanglement is non-increasing only on average, this allows the possibility of obtaining an entangled state with a greater amount of entanglement than initially.
\addtocounter{pprop}{2}
\begin{pprop}[Full monotonicity]\label{prop3.4}
If $\Lambda(\rho)=\sum_ip_i\rho_i$ then 
	\begin{labeq}{eq3.31}
		E(\rho)\geq \sum_ip_iE(\rho_i)
	\end{labeq}
\end{pprop}
Properties \ref{prop3.1} to \ref{prop3.3} are directly related to the structure of the theory for this reason they should always be satisfied by any candidate entanglement measure. An ideal measure should also possess the following properties. On the one hand it would seem reasonable to demand that the total entanglement of a composite system equals the sum of its parts.
\begin{prop}[Additivity]\label{prop3.5}
	\begin{labeq}{eq3.32}
		E(\rho^{\otimes n})=nE(\rho)
	\end{labeq}
\end{prop}
\begin{pprop}[Full additivity]\label{prop3.6}
	\begin{labeq}{eq3.33}
		E(\rho\otimes\sigma)=E(\rho)+E(\sigma)
	\end{labeq}
\end{pprop}
Finally two states that are close to each other should contain approximately equal amounts of entanglement. Ideally an entanglement measure should satisfy the following strong form of continuity
\begin{prop}[Asymptotic Continuity]\label{prop3.5}
For any two states $\rho$, $\sigma$ such that $\lim_{n\to\infty}\norm{\rho^{\otimes n}-\sigma^{\otimes n}}= 0$
	\begin{labeq}{eq3.34}
		\lim_{n\to\infty}\frac{|E(\rho^{\otimes n})-E(\sigma^{\otimes n})|}{n}= 0
	\end{labeq}
\end{prop}

For the bipartite case considered here, a multitude of entanglement measures exist. For our purposes we will only mention but a few needed for later discussions. For a short review on the theory of entanglement measures covering also multipartite systems with an extensive bibliography the reader is referred to \cite{Plenio2005}.

\paragraph{\textbf{Entropy of entanglement:}} The first measure to appear in the literature \cite{Bennett1996}, it is defined as
\begin{labeq}{eq3.35}
	E_\mathcal{S}(\rho):=S(\tr{X}{\rho}) \quad X=A,B
\end{labeq}
where $S(\rho)=-\tr{}{\rho\log{\rho}}$ is the von-Neumann entropy and $Tr_{\scriptscriptstyle X}$ the partial trace. It is the only measure known to satisfy \crefrange{prop3.1}{prop3.5}. For pure states it is equal to the Shannon entropy of it's Schmidt coefficients.
\begin{labeq}{eq3.36}
	E_\mathcal{S}(\phi)=-\sum_i\lambda_i\log{\lambda_i}
\end{labeq}
According to \cref{eq3.36} the entropy of entanglement of a singlet is equal to one and for a d-dimensional maximally entangled state equal to $\log d$, a direct application of the method of Lagrange multipliers verifies that this is the maximum value that can be attained justifying the use of the term maximally for these states. By convention this is also demanded of any measure.
\begin{prop}[Normalization]\label{prop6}
	\begin{labeq}{eq3.37}
		E(\Psi_d^+)=\log d
	\end{labeq}
\end{prop}
\noindent The unit of entanglement is usually called an \emph{e-bit}.

\paragraph{\textbf{Entanglement of formation:}}
Let $\mathcal{E}=\left\{p_i,\psi_i\right\}$ be an ensemble of pure states such that $\rho=\sum_ip_i\ketbra{\psi_i}{}$. The \emph{entanglement of formation} \cite{Bennett1996a} of $\rho$ is now given by
\begin{labeq}{eq3.38}
	E_F(\rho)=\min_{\mathcal E}\sum_ip_iE_\mathcal{S}(\psi_i)
\end{labeq}
and can be interpreted as the least average amount of singlets needed to prepare $\rho$ using LOCC. It is easy to see that for pure states it is equal to the entropy of entanglement. 

This measure is an example of a \emph{convex roof extension} measure, by construction it satisfies \cref{prop3.1,prop3.2,prop3.3,prop6}, continuity was proven in \cite{Nielsen2000}. It is an open question whether it satisfies additivity or not. \cref{eq3.38} is extremely hard to compute in general because of the optimization procedure involved. For $2\times 2$ systems a closed form exists \cite{Wootters1998}.

\paragraph{\textbf{Relative entropy of entanglement:}} An important measure it is defined as
\begin{labeq}{eq3.39}
	E_R(\rho)=\min_{\sigma\in \mathcal S}S(\rho||\sigma)
\end{labeq}
where $S(\rho||\sigma)=\tr{}{\rho log\rho-\rho log\sigma}$ is the relative entropy \cite{Vedral1997,Vedral1998}. This measure can be thought of as a sort of distance\footnote{Strictly speaking it doesn't satisfy all the properties of a distance function specifically the symmetric property since $S(\rho||\sigma)\neq S(\sigma||\rho)$.} of the entangled state to the set of separable states $\mathcal S$ (see \cref{fig3.1}). The relative entropy of entanglement satisfies \cref{prop3.1,prop3.2,prop3.3,prop6}, continuity was proven in \cite{Donald1999}.
\begin{figure}
\begin{center}
\begin{tikzpicture}
	\draw (0,0) ellipse(4cm and 2cm);
	\draw (2.5,0.5)node{$\mathcal S$} (-1,1.5)node{$\mathcal{D_\mathcal{H}}$};
	\draw (2,0) circle(1cm);
		\begin{scope}[shift={(2,0)}]
		\fill (190:1)node[right]{$\sigma$}circle(1pt) (190:4)node[left]{$\rho$}circle(1pt);
		\draw[decorate,decoration={brace,mirror}] (190:1.1)--(190:3.9) node [midway,yshift=0.4cm] {$E_R(\rho)$};
		\end{scope}
\end{tikzpicture}
\end{center}
\caption{\textbf{The Relative entropy of entanglement}. This is a special case of \emph{distance measures}. It can be interpreted as giving the distance of the entangled state from the convex set of separable states $\mathcal S$.}
\label{fig3.1}
\end{figure}
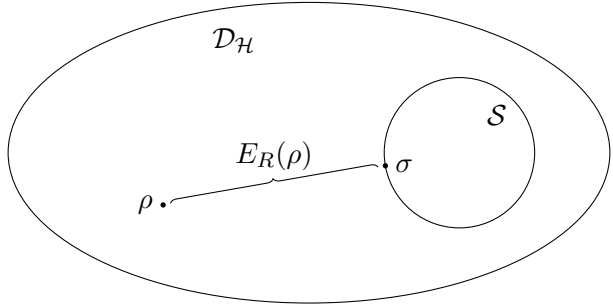
\paragraph{} The last two measures are operational connected to the distillation and dilution protocols.
\paragraph{\textbf{Distillable entanglement:}} Let $\{\Lambda_n\}$ be a sequence of LOCC operations mapping $n$ copies of $\rho$ to $nE-o(n)$ singlets with $\lim_{n\to\infty}o(n)=0$, such a sequence constitutes a distillation protocol. The distillable entanglement is defined by
\begin{labeq}{eq3.41}
	E_D(\rho)=\sup\Set{E|\lim_{n\to \infty}\left[\inf_{\{\Lambda_n\}}\norm{\Lambda_n(\rho^{\otimes n})-P_s^{\otimes nE-o(n)}}\right]\to 0}
\end{labeq}
where $P_s=\ketbra{\psi_s}{}$ is the projection operator for the singlet state. This measure gives the maximum number of singlets per input state that can be asymptotically distilled.

\paragraph{\textbf{Entanglement cost:}} A measure dual to $E_D$ with a similar definition. Let $\{\Lambda_n\}$ be a sequence of LOCC operations mapping $nE+o(n)$ copies of a singlet  to $n$ copies of state $\rho$ with $\lim_{n\to\infty} o(n)=0$, such a sequence constitutes a distillation protocol. The entanglement cost is given by
\begin{labeq}{eq3.42}
	E_C(\rho)=\inf\Set{E|\lim_{n\to \infty}\left[\inf_{\{\Lambda_n\}}\norm{\Lambda_n(P_s^{\otimes nE+o(n)})-\rho^{\otimes n}}\right]\to 0}
\end{labeq}
and gives the minimum amount of singlets per output state that are needed to create it.
\paragraph{}Any subadditive measure $E$ satisfying the continuity property lies between these two \cite{Horodecki2000}
\begin{labeq}{eq3.43}
	E_D(\rho)\leq E(\rho)\leq E_C(\rho)
\end{labeq}
this is extremely useful since such a measure automatically provides an upper bound to the distillable entanglement for which no closed form currently exists.

From any entanglement measure it is possible to construct a regularized version\footnote{Such a measure satisfies by definition the additivity property.}
\begin{labeq}{eq3.44}
	E^\infty(\rho)=\lim_{n\to\infty}\frac{E(\rho^{\otimes n})}{n}
\end{labeq}
An important example of such a measure is the regularized entanglement of formation which is equal to the entanglement cost \cite{Hayden2001}.
\begin{labeq}{eq3.45}
	E_F^\infty(\rho)=E_C(\rho)
\end{labeq}
\section{Entanglement vs Thermodynamics.}\label{chap4}
The non-increase of entanglement under LOCC stated in \cref{prop3.3} bears a strong resemblance to the second law of thermodynamics. This suggests the two theories might be similar in structure. A theory of thermodynamics in which entropy decreases during a process is simply the dual of that presented in \cref{chap2} where $\to$ is replaced everywhere by $\leftarrow$. Physically this is equivalent to reversing the direction of the arrow of time. It is therefore tempting at first sight to make the claim that entanglement is dual to classical thermodynamics and that the thermodynamic and entangled time arrows have opposite orientations. To prove this we have to check whether it is possible to construct an ideal form of entanglement dual to the one given for thermodynamics.

The primitive terms of the theory are easy to identify. A \emph{state} is defined as a positive-semidefinite hermitian operator with unit trace acting on the product Hilbert space $\hilb{A}\otimes\hilb{B}$. \emph{Addition} of states is defined by taking their tensor product. Finally a \emph{natural process} $\to$ between states is mathematically described by a LOCC operation. As with thermodynamics we will only be concerned with \emph{deterministic processes}, references \cite{Lo1997,Vidal1999a,Jonathan1999,Vidal2000,Hardy1999,Vidal1999} adress the problem of stochastic and approximate processes. If state $\rho$ is transformed into $\sigma$ by a \emph{deterministic} LOCC we shall write $\rho\to\sigma$\footnote{There should be no confusion in using the same arrow notation for entanglement as long as one keeps in mind that any Axiom, Theorem, Definition etc. presented here is actually the dual of the corresponding one appearing in \cref{chap2}.}.

To complete the ideal formulation of entanglement based on that of thermodynamics the following dual Axioms must be satisfied.
\begin{axm}\label{axm4.1}
	Addition of states is associative and commutative.
\end{axm}
\begin{axm}\label{axm4.2}
	$\rho\to\rho$
\end{axm}
\begin{axm}\label{axm4.3}
	$\rho\to\sigma\quad\&\quad\sigma\to\tau\implies\rho\to\tau$
\end{axm}
\begin{axm}\label{axm4.4}
	$\rho\to\sigma\iff\rho\otimes\tau\to\sigma\otimes\tau$
\end{axm}
\addtocounter{axm}{1}
\setcounter{subaxm}{0}
\begin{subaxm}\label{axm4.5a}
	$\rho\to\sigma\quad\&\quad\rho\to\tau\implies\sigma\to\tau\quad\text{or}\quad\tau\to\sigma$
\end{subaxm}
\begin{subaxm}\label{axm4.5b}
	$\sigma\to\rho\quad\&\quad\tau\to\rho\implies\sigma\to\tau\quad\text{or}\quad\tau\to\sigma$
\end{subaxm}
\begin{axm}\label{axm4.6}
There exists an internal state e.
\end{axm}
\begin{axm}\label{axm4.7}
	If there exist states $\tau_1$ and $\tau_2$ such that $\rho^{\otimes n}\otimes\tau_1\to\sigma^{\otimes n}\otimes\tau_2$ holds for arbitrarily large positive integers $n$ then $\rho\to\sigma$.
\end{axm}
\begin{axm}\label{axm4.8}
	For any state $\rho$ there exists an equilibrium state $\mu$ such that $\rho\to\mu$. If $\mu$ and $\nu$ are two such states then so is $\mu\otimes\nu$.
\end{axm}

If these Axioms are satisfied by the set of entangled states then according to the discussion at the end of \cref{chap2} there would exist an entropy function on this set which would become totally ordered with respect to adiabatic transitions. The second law for entanglement would thus read 
\begin{labeq}{eq4.1}
	\rho\prec\sigma\quad\text{if and only if}\quad S(\rho)\geq S(\sigma)
\end{labeq}
This would prove to be extremely useful in applications of quantum information theory as it would be possible to determine which processes we are allowed to implement with certainty by simply comparing the entropies of the final and initial states of a possible process. Moreover such a function would serve the role of an entanglement measure since by definition it would satisfy all of the required properties.

In this section we shall discuss the entanglement dynamics of pure states of finite dimensional systems. After introducing the main theorem for deterministic LOCC transitions between states and some corollaries concerning the behavior of maximally entangled states we shall apply these to check whether \crefrange{axm4.1}{axm4.8} are satisfied by the simplest systems possible namely $2\times d$ systems. We shall then repeat the same discussion for systems of any dimension only to demonstrate with the use of counter-examples which Axioms are violated in this case.

\subsection{Entanglement dynamics of pure states.}\label{sec4.1}
A very useful tool in the theory of pure state entanglement is Nielsen's necessary and sufficient conditions for deterministic LOCC transitions \cite{Nielsen1999}.
\begin{thrm}\label{thrm4.1}
If $\phi$ and $\psi$ are two pure states and $\osc\phi=(\alpha_1,\alpha_2,\ldots,\alpha_d)$,  $\osc\psi=(\beta_1,\beta_2,\ldots,\beta_d)$ the vectors of their respective ordered Schmidt coefficients (OSC's) then $\phi\to\psi$ with certainty under LOCC if and only if
	\begin{labeq}{eq4.2}
		\osc\phi\prec\osc\psi
	\end{labeq}
\end{thrm}
 \cref{eq4.2} is known as a majorization relation. Specifically we say that $\osc\phi$ is majorized by $\osc\psi$ if for every $k=1,\ldots,d$ 
\begin{labeq}{eq4.3}
	\sum_{i=1}^k\alpha_i\leq\sum_{i=1}^k\beta_i
\end{labeq}
with equality holding for $k=d$. If $\psi$ has less Schmidt coefficients than $\phi$ we simply add enough coefficients with zero value until their Schmidt ranks become equal.

\cref{thrm4.1} induces an ordering on the set of pure states as well as an equivalence relation, two states $\phi$ and $\psi$ are said to be \emph{comparable} if either $\phi\to\psi$ or $\psi\to\phi$ and \emph{incomparable} otherwise, furthermore if both have the same Schmidt coefficients then we say they are \emph{equivalent} in this case we write $\phi\sim\psi$. By \cref{thrm4.1} if $\phi$ and $\psi$ are two equivalent states then $\phi\to\psi$ and $\psi\to\phi$ both hold.	

The following corollaries dictate the dynamics of maximally entangled states under LOCC.
\begin{cor}\label{cor4.1}
A $d$-dimensional maximally entangled state can always be transformed under LOCC into a $d'$-dimensional maximally entangled state if and only if
	\begin{labeq}{eq4.4}
		d \geq\ d'
	\end{labeq}	
\end{cor}
\begin{proof}
	The proof follows immediately from \cref{thrm4.1}.
\end{proof}
\begin{cor}\label{cor4.2}
	Any pure state $\phi$ of Schmidt rank $d$ can always be obtained under LOCC by a d-dimensional maximally entangled state.
\end{cor}
\begin{proof}
	Suppose that $\alpha_i$, $i=1,2\ldots d$ are the ordered Schmidt coefficients of $\phi$. It can easily be verified that $k(\alpha_1+\alpha_2+\ldots+\alpha_d)\leq d(\alpha_1+\alpha_2+\ldots+\alpha_k)$ which holds for every $1 \leq k\leq d$, since $\sum_i\alpha_i=1$ this can be rewritten as 
\begin{labeq}{eq4.5}
	\frac{k}{d}\leq\alpha_1+\alpha_2+\ldots+\alpha_k
\end{labeq}
 but the left hand side of \cref{eq4.5} is simply the sum of the first $k$ Schmidt coefficients of a d-dimensional maximally entangled state thus proving the corollary.
\end{proof}
\begin{cor}\label{cor4.3}
Given any pure state $\phi$ it is possible to obtain a d-dimensional maximally entangled state using LOCC if and only if
	\begin{labeq}{eq4.6}
		d\leq\alpha_1^{-1}
	\end{labeq}
where $\alpha_1$ is the state's largest Schmidt coefficient.
\end{cor}
\begin{proof}
If  $\phi\to\Psi_d^+$ holds then $\alpha_1\leq d^{-1}$ follows immediately.

Suppose now that $d'$ is the Schmidt rank of $\phi$ and $d\leq\alpha_1^{-1}$. According to \cref{cor4.2} $\alpha_1^{-1}\leq d'$ so $d'\geq d$ follows. Since $\alpha_1\geq\alpha_2\ldots\geq\alpha_{d'}$	  
	\begin{labeq}{eq4.7}
		\sum_{i=1}^k\alpha_i\leq 
			\begin{cases}
				k\alpha_1\leq\frac{k}{d}&\text{for $1\leq k\leq d$}\\
				1&\text{for $d+1\leq k\leq d'$}
			\end{cases}
	\end{labeq}
so according to \cref{thrm4.1} $\phi\to\Psi_d^+$.
\end{proof}
\subsubsection{States of $2\times d$ systems.}\label{sub4.1.1}
The set of pure states of $2\times d$ systems becomes totally ordered under the ordering induced by the majorization relation. This is easy to show by noting that any such state has only two Schmidt coefficients, since their sum is always equal to one we are left with only one independent parameter. Taking this parameter to be the largest Schmidt coefficient of the state, direct application of \cref{thrm4.1} immediately implies that it is always possible to determine, given any pair of states, which one can be transformed to the other under LOCC by simply comparing their first OSC. Furthermore it can be proven that \crefrange{axm4.1}{axm4.8} are satisfied since by the same reasoning most of them simply reduce to statements between real numbers which are trivially satisfied. 

\cref{axm4.1} holds in general for any system since the vector of ordered Schmidt coefficients doesn't depend on the order in which states are added while \cref{axm4.2} is trivially satisfied. We continue by proving that the remaining Axioms are also true.
\begin{proof}[\cref{axm4.3}]
Suppose $\alpha_1$, $\beta_1$, $\gamma_1$ are the first OSC's of $\phi_1$, $\psi$ and $\phi_2$ respectively then $\phi_1\to\psi$ is equivalent to $\alpha_1\leq\beta_1$ and $\psi\to\phi_2$ equivalent to $\beta_1\leq\gamma_1$ it follows that $\alpha_1\leq\gamma_1$ or $\phi_1\to\phi_2$.
\end{proof}
\begin{proof}[\cref{axm4.4}] Suppose $(\alpha_1,\alpha_2)$, $(\alpha'_1,\alpha'_2)$ and $(\beta_1,\beta_2)$ are the OSC's of $\phi_1$, $\phi_2$ and $\psi$ respectively. The largest and smallest Schmidt coefficients of $\phi_1\otimes\psi$ are given by $\alpha_1\beta_1$ and $\alpha_2\beta_2$ and for $\phi_2\otimes\psi$ by $\alpha'_1\beta_1$ and $\alpha'_2\beta_2$ respectively. If $\phi_1\otimes\psi\to\phi_2\otimes\psi$ then according to \cref{thrm4.1} $\alpha_1\beta_1\leq\alpha'_1\beta_1$ from which $\alpha_1\leq\alpha_1'$ so $\phi_1\to\phi_2$ follows from the discussion in the beginning of this section. 

To prove the reverse statement we note that the inequalities in \cref{eq4.3} are trivially satisfied for $k=1$ ($\alpha_1\beta_1\leq\alpha'_1\beta_1$), $k=3$ ($1-\alpha_2\beta_2\leq1-\alpha'_2\beta_2$) and $k=4$ ($1=1$) so only the $k=2$ inequality needs to be checked, the following cases must now be considered.
	\begin{asparaenum}[i)]
		\item{$\alpha_1\beta_2\geq\alpha_2\beta_1$:} Since $\phi_1\to\phi_2$ $\alpha_1\leq\alpha'_1$ and $\alpha'_2\leq\alpha_2$ must hold, it is easy then to show that $\alpha'_1\beta_2\geq\alpha'_2\beta_1$ so $\alpha_1\beta_1+\alpha_1\beta_2\leq\alpha'_1\beta_1+\alpha'_1\beta_2$.
		\item{$\alpha_2\beta_1\geq\alpha_1\beta_2\quad\&\quad\alpha'_2\beta_1\geq\alpha'_1\beta_2$:} In this case $\alpha_1\beta_1+\alpha_2\beta_1=\alpha'_1\beta_1+\alpha'_2\beta_1$.
		\item{$\alpha_2\beta_1\geq\alpha_1\beta_2\quad\&\quad\alpha'_1\beta_2\geq\alpha'_2\beta_1$:} Adding $\alpha'_1\beta_1$ on both sides of the second inequality results in $\beta_1\leq\alpha'_1$ but $\beta_1=\alpha_1\beta_1+\alpha_2\beta_1$ and $\alpha'_1=\alpha'_1\beta_1+\alpha'_1\beta_2$.
	\end{asparaenum}
	
In either case $\osc{\phi_1\otimes\psi}\prec\osc{\phi_2\otimes\psi}$ so ${\phi_1\otimes\psi\to\phi_2\otimes\psi}$ follows.\qedhere
\end{proof}
\begin{proof}[\cref{axm4.5a}]
	Suppose $\alpha_1$, $\alpha'_1$ and $\beta_1$ are the highest OSC's of $\phi_1$, $\phi_2$ and $\psi$ respectively then $\psi\to\phi_1\implies\beta_1\leq\alpha_1$ and $\psi\to\phi_2\implies\beta_1\leq\alpha'_1$ so either $\alpha_1\leq\alpha'_1$ or $\alpha'_1\leq\alpha_1$ holds which implies that either $\phi_1\to\phi_2$ or $\phi_2\to\phi_1$ respectively. \cref{axm4.5b} is proven similarly.
\end{proof}
\begin{proof}[\cref{axm4.6}]
The internal state in this case can be chosen to be the 2-dimensional maximally entangled state $\Psi_2^+$. To prove this suppose $\alpha_1$ and $x_1$ are the highest OSC's of $\phi$ and $\xi$ respectively then by \cref{cor4.2} there always exist positive integers $n$ and $m$ with $nm\geq n+1$ such that $\Psi_2^{+\otimes nm}\to\phi^{\otimes n}\otimes\xi$ and similarly by \cref{cor4.3} positive integers $n,m$ with $\alpha_1^nx_1\leq2^{-nm}$ such that $\phi^{\otimes n}\otimes\xi\to\Psi_2^{+\otimes nm}$ 
\end{proof}
\begin{proof}[\cref{axm4.7}]
Suppose that $\phi_1^{\otimes n}\otimes\eta\to\phi_2^{\otimes n}\otimes\chi$ holds for arbitrarily large positive integers $n$. If $\alpha_1$, $\alpha'_1$, $\eta_1$, $\chi_1$ are the largest OSC's of $\phi_1$, $\phi_2$, $\eta$ and $\chi$ respectively then by \cref{thrm4.1} $\alpha_1^n\eta_1\leq(\alpha'_1)^n\chi_1$ taking the logarithm of each side and dividing by $n$ results in $\log\alpha_1+\frac{1}{n}\log\eta_1\leq\log\alpha'_1+\frac{1}{n}\log\chi_1$ which holds for arbitrarily large values of $n$, taking the limit $n\to\infty$ $\alpha_1\leq\alpha'_1$ follows and $\phi_1\to\phi_2$.
\end{proof}
\begin{proof}[\cref{axm4.8}]
It is natural to choose the set of separable states $\mathcal S$ as the equilibrium states of the theory, for pure states this set consists of product states. The tensor product of two product states is also a product state.
\end{proof}

At this point it would seem that for the special case of $2\times d$ systems the theory of entanglement manipulations under deterministic LOCC is formally equivalent to the classical theory of thermodynamics since the primitive terms of both theories satisfy the same Axioms. If that were the case then the existence of a positive non-increasing entropy function of state would automatically be guaranteed and could be defined as giving the amount of entanglement present in that state. The first step in quantifying it would be to choose an appropriate reference state using it's entanglement as the unit of measure, for our purposes the 2-dimensional maximally entangled state $\Psi_2^+$ could be given that role. The next step would then be to define for any state $\phi$ the following sets 
\begin{labeq}{eq4.8}
	U(\phi):=\Set{\frac{m}{n}|\Psi_2^{+\otimes m}\to\phi^{\otimes n}}
\end{labeq}
\begin{labeq}{eq4.9}
	L(\phi):=\Set{\frac{m}{n}|\phi^{\otimes n}\to\Psi_2^{+\otimes m}}
\end{labeq}
and then define the entanglement of $\phi$ as the Dedekind cut of these two sets. The existence of such a unique number would imply that there exists a reversible process between multiple copies of $\phi$ and $\Psi_2^+$. Unfortunately this is incorrect.

Evidence to the contrary is given by the proof that no matter what reference state we choose no such reversible processes exist apart from trivial cases.
\begin{thrm}\label{thrm4.2}
There exist no positive integers $n$ and $m$ such that for any two $2\times d$ dimensional states $\phi$ and $\psi$
	\begin{labeq}{eq4.10}
		\phi^{\otimes n}\sim\psi^{\otimes m}
	\end{labeq}
\end{thrm}
\begin{proof}
Let $(\alpha_1,\alpha_2)$ and $(\alpha'_1,\alpha'_2)$ be the OSC's of $\phi$ and $\psi$ respectively. If $\alpha_1=\alpha'_1$ then $\phi$ and $\psi$ are the same state and \cref{eq4.10} holds for $m=n$ which according to \cref{axm4.2} is trivial. Assume then without loss of generality that $\alpha_1<\alpha'_1$. According to \cref{thrm4.1} $\alpha_1^n={\alpha'_1}^m$ so $\frac{m}{n}=\frac{\log\alpha_1}{\log{\alpha'_1}}>1$ on the other hand $\alpha_2^n={\alpha'_2}^m$ so $\frac{m}{n}=\frac{\log\alpha_2}{\log{\alpha'_2}}<1$ resulting in a contradiction.
\end{proof}
\begin{cor}\label{cor4.1}
The set of all $2\times d$ dimensional pure states and their possible tensor products is partially ordered.
\end{cor}
\begin{proof}
Suppose $U(\phi|\psi)$ and $L(\phi|\psi)$ are the same sets as those in Equations \eqref{eq4.8} and \eqref{eq4.9} where $\Psi_2^+$ is replaced by $\psi$. According to \cref{thrm4.2} ${\max L(\phi|\psi)<\min U(\phi|\psi)}$. Since it is always possible to find positive integers $n$ and $m$ such that $\max L(\phi|\psi)<\frac{m}{n}<\min U(\phi|\psi)$ this implies that $\psi^{\otimes m}\not\to\phi^{\otimes n}$ and $\phi^{\otimes n}\not\to\psi^{\otimes m}$ making $\phi^{\otimes n}$ and $\psi^{\otimes m}$ incomparable.
\end{proof}
Notice that introducing adiabatic processes by considering transitions between states which are the tensor product of an entangled state and a disentangled one cannot restore the total ordering as was possible for the classical theory of thermodynamics since such tensor products have no effect on the state's ordered Schmidt coefficients\footnote{Each Schmidt coefficient is simply multiplied by one.}.

The controversy can be clarified by noting that \crefrange{axm4.1}{axm4.8} must hold for tensor the product of any  $2\times d$ dimensional entangled states and not only for single copies of such states. Since $n$ tensor products of $2\times d$ states can simply be interpreted as a $2^n\times d^n$ entangled pure state it suffices to check whether single copies of higher dimensional states satisfy the Axioms.
\subsubsection{States of any dimension.}\label{sub4.1.2}
As was shown in \cref{cor4.1} the set of $2\times d$ pure entangled states is only partially ordered. This is also true for the set of entangled pure states. It is easy to construct many examples where the conditions in \cref{thrm4.1} are violated so neither state can be transformed deterministically into the other under LOCC. One such example is given by the following pair of incomparable states
\begin{align}
	\psi_1&=\sqrt{0.4}\ket{00}+\sqrt{0.4}\ket{11}+\sqrt{0.1}\ket{22}+\sqrt{0.1}\ket{33}\label{eq4.11}\\
	\psi_2&=\sqrt{0.5}\ket{00}+\sqrt{0.25}\ket{11}+\sqrt{0.25}\ket{22}\label{eq4.12}
\end{align}
In \cite{Jonathan1999a} it was shown that there exists a state $\eta=\sqrt{0.6}\ket{00}+\sqrt{0.4}\ket{11}$ such that 
\begin{labeq}{eq4.13}
	\psi_1\otimes\eta\to\psi_2\otimes\eta
\end{labeq}
is now possible. This \emph{catalytic process} violates \cref{axm4.4} for now \cref{eq4.13} does not necessarily imply that $\psi_1\to\psi_2$. A possible solution might be to slightly modify the definition of accessibility in a way similar to \cref{def2.1} by demanding that two states are comparable if a catalytic process between them exists. Even with such an extension the ordering remains partial. To see this suppose $\phi$ and $\psi$ are two pure states of equal Schmidt rank such that $\alpha_1\leq\beta_1$ and $\alpha_d\leq\beta_d$ where $\alpha_1,\alpha_d$ are the first and last OSC's of $\phi$ and $\beta_1,\beta_d$ the corresponding coefficients for $\psi$. No matter which state is chosen as a catalyst these inequalities will remain unchanged, there is therefore no catalytic process between $\phi$ and $\psi$.

\cref{axm4.5a,axm4.5b} are also invalid according to the following theorem.
\begin{thrm}\label{thrm4.3}
For any two incomparable states $\phi_1$, $\phi_2$ there always exist states $\psi$ and $\psi'$ such that
\begin{labeq}{eq4.14}
	\psi\to\phi_1\quad\&\quad\psi\to\phi_2
\end{labeq}
and
\begin{labeq}{eq4.15}
	\phi_1\to\psi'\quad\&\quad\phi_2\to\psi'
\end{labeq}
\end{thrm}
\begin{proof}
Let $(\alpha_i)_{i=1}^{d_1}$, $(\beta_i)_{i=1}^{d_2}$  be the OSC's of $\phi_1$ and $\phi_2$ with Schmidt ranks $d_1$ and $d_2$ respectively. It is always possible to construct from these states $d\geq\max(d_1,d_2)$ ordered Schmidt coefficients $(\gamma_i)_{i=1}^d$ satisfying 
\begin{labeq}{eq4.16}
	\sum_{i=1}^k\gamma_i\leq\min{\(\sum_{i=1}^k\alpha_i,\sum_{i=1}^k\beta_i\)}
\end{labeq}
for every $1\leq k\leq d$ and $d'\leq\min{\(d_1,d_2\)}$ coefficients $(\delta_i)_{i=1}^{d'}$ satisfying
\begin{labeq}{eq4.17}
\max{\(\sum_{i=1}^k\alpha_i,\sum_{i=1}^k\beta_i\)}\leq\sum_{i=1}^k\delta_i
\end{labeq}
for every $1\leq k\leq d'$. If $\psi$ and $\psi'$ are the states with OSC's $(\gamma_i)$ and $(\delta_i)$ respectively then \cref{eq4.14} and \cref{eq4.15} follow. 
\end{proof}
For example it can easily be shown that for the two states given in \cref{eq4.11,eq4.12} $\psi=\frac{27}{80}\ket{00}+\frac{27}{80}\ket{11}+\frac{9}{40}\ket{22}+\frac{1}{10}\ket{33}$ and $\psi'=\frac{16}{27}\ket{00}+\frac{8}{27}\ket{11}+\frac{1}{10}\ket{22}$.

With the help of the following counterexample it can also be shown that \cref{axm4.7} cannot hold in general. Specifically suppose that $\phi$ and $\psi$ are two incomparable states and $\eta$ is an appropriate entanglement `catalyst' for this pair. Employing the commutativity property of the tensor product it is easy to show that 
\begin{labeq}{eq4.18}
	\phi^{\otimes n}\otimes\eta\to\psi^{\otimes n}\otimes\eta
\end{labeq}
\cref{eq4.18} is valid for arbitrarily large positive integers $n$ but $\phi\not\to\psi$ and $\psi\not\to\phi$ by assumption.

The existence of a positive non-increasing entropy function depends on this Axiom to be true for any pure entangled state, we have therefore demonstrated that the theory of pure state entanglement is formally inequivalent to that of classical thermodynamics. Nonetheless a weak form of the second law can be recovered in this case with the help of results from majorization theory. It can be proven that whenever \cref{eq4.2} holds for any two states there exist functions of the Schmidt coefficients such that 
\begin{labeq}{eq4.19}
	f(\phi)\geq f(\psi)
\end{labeq}
holds. These functions are known as \emph{entanglement monotones} and can be used as an entanglement measure \cite{Vidal2000a}, an example of such a function is the Shannon entropy of the Schmidt coefficients given by \cref{eq3.36}.

We conclude this section with a proof based on \cref{thrm4.1} that distillation and dilution of entanglement are irreversible.
\begin{thrm}\label{thrm4.4}
There exist no positive integers $n$ and $m$ such that
	\begin{labeq}{eq4.20}
		\phi^{\otimes n}\sim\Psi_d^{+\otimes m}
	\end{labeq}
\end{thrm}
\begin{proof}
The proof is similar to that of \cref{thrm4.2}. Suppose $\alpha_1$ and $\alpha_{d'}$ are the largest and smallest OSC's of $\phi$ where $d'$ is the state's Schmidt rank. According to \cref{thrm4.1} \cref{eq4.20} implies that $\alpha_1^n=d^{-m}$ and ${\alpha_{d'}^n=d^{-m}}$, it is easy to see that if $\alpha_1=\alpha_{d'}$ then $\phi$ must be a $d'$ dimensional maximally entangled state and \cref{eq4.20} holds trivially since then $d'^n=d^m$. For $\alpha_1>\alpha_{d'}$ a contradiction arises.
\end{proof}
\section{Thermodynamics of Entanglement.}\label{chap5}
So far it has been demonstrated that the theory of entanglement manipulations under deterministic LOCC transformations is formally inequivalent to that of classical thermodynamics when a finite number of entangled states is considered. This suggests that the class of LOCC is to restrictive. Moreover it can be shown that entanglement distillation of mixed states requires collective measurements to be carried out on a number of copies of the initial state, the greater this number the more successful the distillation \cite{Linden1998,Kent1998}. It might therefore be possible to construct a complete analogy between entanglement and thermodynamics by lifting some of the restrictions in order to allow a larger set of processes to become possible. This is also suggested by an apparent phenomenological similarity between the two theories. 

We begin this section by describing such an analogy between entanglement and thermodynamics. We will then demonstrate how the theory of pure state entanglement can be modified in such a way such that it possesses a thermodynamical structure. Concluding we shall discuss why the same approach does not work in the case of mixed states and how the strong form of the second law can nonetheless be derived by a different method.
\subsection{Thermodynamical analogies.}\label{sec4.2}
Mixed state entanglement transformations exhibit irreversibility, to make this clear it suffices to look at the distillation and dilution protocols for mixed states. On the one hand it is known that any $2\times d$ dimensional mixed state is distillable \cite{Horodecki1997c,Dur2000a}, on the other hand it has been proven that although bound entangled states require a non-zero amount of singlets to be formed it is nonetheless impossible to distill them  no matter how many copies we start with\cite{Vidal2001,Yang2005}, specifically if $\rho_b$ is a bound entangled state then
\begin{labeq}{eq4.26}
	0=E_D(\rho_b)<E_C(\rho_b)
\end{labeq}
This is the simplest example of an irreversible process. In \cite{Horodecki2002} it was argued that such irreversibility is also inherent in thermodynamics, specifically it is an immediate consequence of the second law that although useful work can be used in order to create a given heat bath, it is nonetheless impossible to recover it. This analogy was used to suggest that \emph{bound entanglement} is the analogue of \emph{heat}, in this context the \emph{entanglement cost} of an entangled state would be the analogue of \emph{internal energy}, and \emph{maximally entangled states} the analogue of \emph{useful work}. Distillation of entanglement would then consist of the following process
\begin{labeq}{eq4.21}
	\rho^{\otimes n}\to\psi_s^{\otimes nE_D}\otimes\rho_g
\end{labeq}
where $n$ copies of state $\rho$ are distilled into $nE_D(\rho)$ copies of singlets plus a leftover garbage state  $\rho_g$ containing an amount $E_b(\rho)$ of bound entanglement. Entanglement dilution would similarly be defined by
\begin{labeq}{eq4.22}
	\psi_s^{\otimes nE_D}\otimes\rho_g\to\rho^{\otimes n}
\end{labeq}
Supposing $E_C$ is an additive measure such that $E_C(\psi_s)=1$ and $E_C(\rho_g)=E_b$ \cref{eq4.21,eq4.22} would imply that 
\begin{labeq}{eq4.23}
E_C=E_D+E_b
\end{labeq} 
\cref{eq4.23} would be the entanglement analogue of the \emph{first law}, for a formulation of this law in the context of quantum information theory see also \cite{Horodecki1998a}.
  
By making use of the analogy any equation in thermodynamics would have a corresponding entangled version. As an example consider the entanglement analogue of the Gibbs-Helmholtz equation
\begin{labeq}{eq4.24}
	E_D=E_C-T_eS_e
\end{labeq}
where $T_e$ denotes the entanglement temperature and $S_e$ is a suitable entropy function  \cite{Horodecki1998}. According to \cref{eq4.24} the temperature of a given entangled state would then by given by 
\begin{labeq}{eq4.25}
	T_e=\frac{E_C-E_D}{S_e}
\end{labeq}

Unfortunately it seems that the distillation and dilution processes described in \cref{eq4.21,eq4.22} are irreversible as was demonstrated in \cite{Horodecki2002} where a counterexample was given. Moreover making the reasonable assumption that the usual von-Neumann entropy is a good candidate for $S_e$ \cref{eq4.25} would imply that any pure entangled state has zero temperature while the temperature of a separable state would be infinite.
\subsection{Thermodynamics of pure states.}\label{sec5.1}
In \cref{sec3.3} it was discussed how distillation and dilution of entanglement are mutually reversible processes in the limit of an asymptotic number of copies of the initial state. This is because $n$ copies of any pure state $\phi$ are a good approximation to $nE_\mathcal{S}-o(n)$ singlets, where $o(n)$ is an asymptotically vanishing real number. Using the fidelity $F(\phi,\psi)$ as a measure of the degree of similarity between two states \cite{Jozsa1994a} this can be formalized as
\begin{labeq}{eq5.1}
	\lim_{n\to\infty} F(\phi^{\otimes n},\psi_s^{\otimes nE_\mathcal{S}(\phi)-o(n)})=1
\end{labeq}
This reversibility was employed by Popescu and Rohrlich in order to construct an analogy between pure state entanglement and thermodynamics \cite{Popescu1996}. By making some simple assumptions it was pointed out that in this limit the only entanglement measure that exists is the entropy of entanglement $E_\mathcal{S}$. Under this approach an argument resembling the one about Carnot heat engines can be given which clearly demonstrates that entanglement is a non-increasing resource under LOCC\cite{Plenio1998a}. 

Specifically suppose there exists a distillation process which achieves a greater efficiency than the one given in \cref{eq3.26} which for a general pure state is equal to the entropy of entanglement. By using the more efficient method it would then be possible to extract more singlets out of $n$ copies of $\phi$ than with the one described in \cref{sec3.3}. Diluting these singlets we will then end up with a greater number of copies of $\phi$ than what we started with. Repeating the same process it is easy to see that after each distillation step a greater number of singlets than previously will be produced which can then be diluted into a greater number of copies of $\phi$ from which an even greater amount of singlets can be distilled \textit{ad infinitum}. If this were possible then we would be able to obtain an infinite amount of singlets from just a finite number of less entangled states essentially for free, the only solution to this problem is for entanglement to be non-increasing under LOCC much the same as thermodynamical entropy needs to be non-decreasing in order to avoid the existence of perpetual motion machines of the second kind.

In order to give a rigorous treatment of the thermodynamics of asymptotic entanglement manipulations we need to extend the definition of accessibility between states accordingly. This is usually done with the help of approximate transformations. The idea behind such an approach is that we consider the best LOCC transformation that maps $n$ copies of $\phi$ as close as possible to $n$ copies of $\psi$, if $\psi$ is \emph{asymptotically accessible} by $\phi$ (written as $\phi\asto\psi$) then in the limit $n\to\infty$ the transformation becomes exact. 
\begin{defin}\label{def5.1}
$\phi\asto\psi$ if and only if there exists a sequence of LOCC transformations $\{\Lambda_n\}$ such that 
	\begin{labeq}{eq5.2}
		\lim_{n\to\infty}\[\inf_{\{\Lambda_n\}}\norm{\Lambda_n(\phi^{\otimes n})-\psi^{\otimes n}}\]=0
	\end{labeq}
\end{defin}

Under a similar definition it was proven in \cite{Vedral2002} that  \crefrange{axm4.1}{axm4.8} are satisfied. The asymptotic theory of pure state entanglement is therefore formally equivalent to the classical theory of thermodynamics.  The entropy of entanglement singles out once more as the unique entanglement measure of the theory. We shall now attempt and do the same by providing an alternative definition of asymptotic accessibility between pure states.

To begin with note that although there exist pure states that are mutually incomparable under deterministic LOCC transformations, like the states given in \cref{eq4.11,eq4.12} for example, it is nonetheless possible to \emph{stochastically} transform one into the other \cite{Vidal1999a},  the probability of success $P(\phi\to\psi)$ for such a stochastic transformation of $\phi$ into $\psi$  is given by
\begin{labeq}{eq5.3}
	P(\phi\to\psi)=\min_{l\in[1,n]}\frac{E_l(\phi)}{E_l(\psi)}
\end{labeq}
where $E_l(\phi)=\sum_{i=l}^n\alpha_i$ and $\alpha_i$ are the OSCs of $\phi$, this probability is nonzero only when the Schmidt rank of the initial state is greater than the Schmidt rank of the final state. We now proceed by giving our definition of \emph{asymptotic accessibility}.
\begin{defin}\label{def5.2}
$\phi\asto\psi$ if and only if
	\begin{labeq}{eq5.4}
		\lim_{n\to\infty}P(\phi^{\otimes n}\to\psi^{\otimes n})= 1
	\end{labeq}
\end{defin}
It can be proven that the set of entangled pure states is totally ordered with respect to the ordering induced by \cref{def5.2}.
\begin{lem}\label{lem5.1} 
For any two entangled pure states $\psi$ and $\phi$ either
	\begin{labeq}{eq5.5}
		\phi\asto\psi\quad\text{or}\quad\psi\asto\phi
	\end{labeq}
\end{lem}
\begin{proof}
Suppose without loss of generality that $\phi$ has a greater number of Schmidt coefficients than $\psi$, we need to calculate the left hand side of \cref{eq5.4} and show that it is equal to one, according to \cref{eq5.1} $n$ copies of any pure entangled state are approximately equal to a maximally entangled state of $2^{nE_\mathcal{S}-o(n)}$ dimension, taking this into account 
\begin{labeq}{eq5.6}
	\lim_{n\to\infty}P(\phi^{\otimes n}\to\psi^{\otimes n})=\lim_{n\to\infty}\left\{\min_l\frac{(2^{nE_\mathcal{S}(\phi)-o(n)}-l+1)}{(2^{nE(\psi)-o'(n)}-l+1)}\frac{2^{-nE_\mathcal{S}(\phi)+o(n)}}{2^{-nE_\mathcal{S}(\psi)+o'(n)}}\right\}
\end{labeq}
For a finite value of $l$ the right hand side of \cref{eq5.6} is equal to one this is also true if $l=2^{nE_\mathcal{S}(\phi)-o(n)}$.\qedhere
\end{proof}

What remains now is to show that \crefrange{axm4.1}{axm4.8} are satisfied. \cref{axm4.1,axm4.2} are trivial so we proceed by giving a proof for the remaining Axioms.
\begin{proof}[\cref{axm4.3}]
Suppose $\phi_1\asto\psi$ and $\psi\asto\phi_2$ and let $p_n=P(\phi_1^{\otimes n}\to\psi^{\otimes n})$ and $q_n=P(\psi^{\otimes n}\to\phi_2^{\otimes n})$, it follows that $P(\phi_1^{\otimes n}\to\phi_2^{\otimes n})=p_nq_n$ this means that $\lim_{n\to\infty}P(\phi_1^{\otimes n}\to\phi_2^{\otimes n})=\lim_{n\to\infty}p_nq_n=\lim_{n\to\infty}p_n\lim_{n\to\infty}q_n=1$.
\end{proof}
\begin{proof}[\cref{axm4.4}]
	\item($\phi_1\asto\phi_2\implies\phi_1\otimes\psi\asto\phi_2\otimes\psi$): If $d_{\phi_1},d_{\phi_2},d_\psi$ are the Schmidt ranks of $\phi_1,\phi_2$ and $\psi$ respectively then by assumption $d_{\phi_1}\geq d_{\phi_2}$ from which $d_{\phi_1}d_\psi\geq d_{\phi_2}d_\psi$. According to \cref{lem5.1} $\phi_1\otimes\psi\asto\phi_2\otimes\psi$.

	\item{($\phi_1\otimes\psi\asto\phi_2\otimes\psi\implies\phi_1\asto\phi_2$):} If $d_{\phi_1},d_{\phi_2},d_\psi$ are the Schmidt ranks of $\phi_1,\phi_2$ and $\psi$ respectively then by assumption $d_{\phi_1}d_\psi\geq d_{\phi_2}d_\psi$ from which $d_{\phi_1}\geq d_{\phi_2}$. According to \cref{lem5.1} $\phi_1\asto\phi_2$.\qedhere
\end{proof}
Since according to \cref{lem5.1} the set of pure entangled states is totally ordered with respect to $\asto$ \cref{axm4.5a,axm4.5b} are trivially satisfied. For the same reason \cref{axm4.6} is also true, it is natural to choose the singlet state as the internal state of the theory.
\begin{proof}[\cref{axm4.7}]
Let $\chi_1$, $\chi_2$ be two pure states such that ${\phi^{\otimes n}\otimes\chi_1\asto\psi^{\otimes n}\otimes\chi_2}$ holds for arbitrarily large values of $n$. If $d_\phi,d_{\chi_1},d_\psi,d_{\chi_2}$ are the Schmidt ranks of each state then $d_{\phi}^nd_{\chi_1}\geq d_{\psi}^nd_{\chi_2}$ must hold for arbitrarily large positive values of $n$ taking the logarithm of both sides, dividing by $n$ and taking the limit $n\to\infty$ results in $d_\phi\geq d_\psi$ so by \cref{lem5.1} $\phi\asto\psi$.
\end{proof}
\begin{proof}[\cref{axm4.8}]
The set of equilibrium states is simply the set of separable pure states, since the ordering is already total there is no need to introduce them.\let\qed\relax
\end{proof}

Following the discussion at the end of \cref{chap2} there exists a positive non-increasing additive function of state $S$. To determine it's value for any pure state we construct the following sets
\begin{labeq}{eq5.7}
	U(\phi):=\Set{\frac{m}{n}|\psi_s^{\otimes m}\asto\phi^{\otimes n}}
\end{labeq}
\begin{labeq}{eq5.8}
	L(\phi):=\Set{\frac{m}{n}|\phi^{\otimes n}\asto\psi_s^{\otimes m}}
\end{labeq}
$S(\phi)$ is now given by the Dedekind cut of these two sets
\begin{labeq}{eq5.9}
	S(\phi)=\sup L(\phi)=\inf U(\phi)
\end{labeq}
By definition this function is additive, non-increasing and asymptotically continuous having a value of $\log d$ for a d-dimensional maximally entangled state, by the uniqueness theorem for entanglement measures \cite{Donald2002} it follows it is equal to the entropy of entanglement $E_\mathcal{S}$. We can now formulate the analogue of the strong form of the second law for pure states.
\begin{figure}
\begin{center}
\begin{tikzpicture}[scale=0.6]
\draw[color=gray] plot[mark=*,mark options={fill=black,draw=black},ycomb] file {plot.txt};
\draw[->,thick] (0,0) -- (0,10.4) node[right] {$E_\mathcal{S}$};
\draw[->,thick] (0,0) -- (10.4,0) node[right] {$\log dim\mathcal{H}$};
\foreach \x in {0,...,10}
     		\draw (\x,1pt) -- (\x,-3pt)
			node[anchor=north] {\x};
    	\foreach \y in {0,...,10}
     		\draw (1pt,\y) -- (-3pt,\y) 
     			node[anchor=east] {\y}; 
\end{tikzpicture}
\end{center}
\caption{\textbf{Set of pure entangled states in thermodynamic space.} Full circles represent maximally entangled states of corresponding dimension, the set of separable states lies on the axis $E_\mathcal{S}=0$.}\label{fig5.1}
\end{figure}
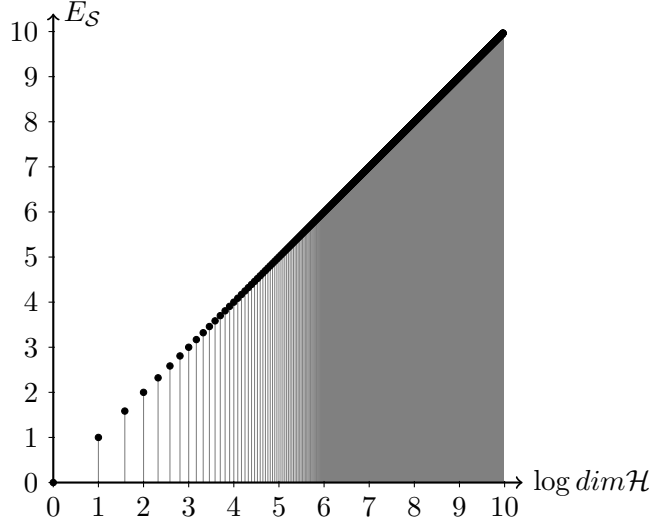
\begin{stat}[\textbf{Second law of pure state entanglement thermodynamics.}]\label{thrm5.2}
	\begin{labeq}{eq5.10}
		\phi\asto\psi\quad\text{if and only if}\quad E_\mathcal{S}(\phi)\geq E_\mathcal{S}(\psi)
	\end{labeq}
\end{stat}

As with classical thermodynamics the set of pure entangled states can be represented in \emph{thermodynamic space} where the entropy of a state is plotted against it's components of content. In this case there is only one component of content given by the logarithm of the dimension of the state's Hilbert space $\log(dim\mathcal{H})$, let's call this the \emph{dimensionality} of the state. Recall that components of content must be additive functions of state a property shared by this function. 

The thermodynamic space for pure state entanglement is pictured in \cref{fig5.1}. This space is discrete becoming continuous in the limit of very large dimensions. The anti-equilibrium surface, comprised of the maximally entangled pure states, lies on the line $E_\mathcal{S}=\log dim(\mathcal{H})$. It's derivative with respect to the dimensionality is equal to one and defines an intensive quantity for the set of anti-equilibrium states which can be thought of as representing a sort of \emph{entanglement density}. The fact that tensor products of maximally entangled states always result in a higher dimensional maximally entangled state can be simply expressed by saying that all such states have the same density, furthermore during any entanglement processes entanglement density decreases. These conditions are reasonable in the context of entanglement theory.
\subsection{Thermodynamics of mixed states.}\label{sec5.2}
Since the asymptotic limit successfully leads to a thermodynamical description for pure states it is natural to ask the question whether this is also possible for mixed states. Unfortunately under asymptotic LOCC transformations the set of mixed states is only partially ordered, this is the reason why so many different entanglement measures exist \cite{Morikoshi2004}. To demonstrate this consider the following two states, one composed of $n$ copies of a bound entangled state $\rho_b$ and the other composed of $m$ singlets, note that $n$ copies of a positively partially transposed matrix also possess the same property so $\rho_b^{\otimes n}$ is also a bound entangled state. In order to create this state we need to dilute at least $nE_C(\rho_b)$ singlets where $E_C$ is the entanglement cost defined in \cref{eq3.42}, if $m<nE_C(\rho_b)$  then
\begin{labeq}{eq5.11}
	P_s^{\otimes m}\nasto\rho_b^{\otimes n}
\end{labeq}
According to \cref{sec3.3} it is impossible to distill a bound entangled state so
\begin{labeq}{eq5.12}
	\rho_b^{\otimes n}\nasto P_s^{\otimes m}
\end{labeq}
these two states are therefore incomparable. 

Different entanglement measures may also give rise to different partial orderings on the set of mixed states \cite{Virmani2000} unless they happen to coincide, i.e. in one measure $E(\rho)\leq E(\sigma)$ while in another $E'(\sigma)\leq E'(\rho)$ might hold so there is no effective way to compare the two states. Irreversibility was also discussed  in \cref{sec4.2} .

In order to attempt to construct the thermodynamics of mixed state entanglement, a further extension of the theory to a larger class $\mathcal O$ of allowed operations such that $LOCC\subseteq\mathcal{O}$  is therefore necessary, \cref{def5.1} is readily generalized in this case. 
\begin{defin}\label{def5.3}
$\rho\exasto{O}\sigma$ if and only if there exists a sequence of transformations $\{\Lambda_n\}\in\mathcal{O}$ such that 
	\begin{labeq}{eq5.13}
		\lim_{n\to\infty}\[\inf_{\{\Lambda_n\}}\norm{\Lambda_n(\rho^{\otimes n})-\sigma^{\otimes n}}\]=0
	\end{labeq}
\end{defin}

It is currently uncertain if a treatment of the thermodynamics of entanglement based on an axiomatic formulation is possible for mixed states as was for pure. This is because the reverse statement of \cref{axm4.4} and \cref{axm4.7} are hard to prove in this case. Nonetheless it is possible to derive the strong form of the second law employing a method based on asymptotic entanglement measures. 

Specifically by generalizing \cref{eq3.43} under a given class of operations $\mathcal O$, entanglement of distillation $E_D^{\mathcal O}$ and entanglement cost $E_C^{\mathcal O}$ defined analogously to \cref{eq3.41,eq3.42} must be extreme measures in the sense that any asymptotically continuous entanglement measure E must lie between the two
\begin{labeq}{eq5.14}
	E_D^{\mathcal O}(\rho)\leq E(\rho) \leq E_C^{\mathcal O}(\rho)
\end{labeq}
If it were possible to prove that $E_D^{\mathcal O}(\rho)=E_C^{\mathcal O}(\rho)$ for every $\rho\in\mathcal D$ then it would immediately follow that a unique measure of entanglement $E$ exists which would be non-increasing under these transformations. A simple example of this approach is the set of LOCC operations for which it is known that ${E_D^{LOCC}(\phi)=E_C^{LOCC}(\phi)}$ for any pure state $\phi$. The strong form of the second law can now be proven as a theorem. 
\begin{thrm}\label{thrm5.2}
For any two states $\rho,\sigma\in\mathcal{D}$ 
\begin{labeq}{eq5.15}
	\rho\exasto{O}\sigma\quad\text{if and only if}\quad E(\rho)\geq E(\sigma)
\end{labeq}
\end{thrm}
\begin{proof}
By construction $E$ is non-increasing so $\rho\exasto{O}\sigma$ immediately implies that $E(\rho)\geq E(\sigma)$. 

To prove the converse let $\{\Lambda_n\}$ be an optimal distillation protocol for $\rho$, $\{\Lambda_n''\}$ an optimal dilution protocol for $\sigma$ and $\{\Lambda_n'\}$ a sequence of transformations in $\mathcal{O}$ such that $$\lim_{n\to\infty}\norm{\Lambda_n'(P_s^{\otimes nE(\rho)-o(n)})-P_s^{\otimes nE(\sigma)+o'(n)}}=0$$ Setting ${\Lambda_n'''=\Lambda_n''\circ\Lambda_n'\circ\Lambda_n}$ it can now be shown that
\begin{align}\label{eq5.16}
	\norm{\Lambda_n'''(\rho^{\otimes n})-\sigma^{\otimes n}}&\leq\norm{\Lambda_n'''(\rho^{\otimes n})-\Lambda_n''(P_s^{\otimes nE(\sigma)+o'(n)})}\nonumber
	\\&+\norm{\Lambda_n''(P_s^{\otimes nE(\sigma)+o'(n)})-\sigma^{\otimes n}}\nonumber
	\\&\leq\norm{\Lambda_n'(\Lambda_n(\rho^{\otimes n}))-P_s^{\otimes nE(\sigma)+o'(n)}}\nonumber
	\\&+\norm{\Lambda_n''(P_s^{\otimes nE(\sigma)+o'(n)})-\sigma^{\otimes n}}\nonumber
	\\&\leq\norm{\Lambda_n'(\Lambda_n(\rho^{\otimes n}))-\Lambda_n'(P_s^{\otimes nE(\rho)-o(n)})}\nonumber
	\\&+\norm{\Lambda_n'(P_s^{\otimes nE(\rho)-o(n)})-P_s^{\otimes nE(\sigma)+o'(n)}}\nonumber
	\\&+\norm{\Lambda_n''(P_s^{\otimes nE(\sigma)+o'(n)})-\sigma^{\otimes n}}\nonumber
	\\&\leq\norm{\Lambda_n(\rho^{\otimes n})-P_s^{\otimes nE(\rho)-o(n)}}\nonumber
	\\&+\norm{\Lambda_n'(P_s^{\otimes nE(\rho)-o(n)})-P_s^{\otimes nE(\sigma)+o'(n)}}\nonumber
	\\&+\norm{\Lambda_n''(P_s^{\otimes nE(\sigma)+o'(n)})-\sigma^{\otimes n}}
\end{align}
\noindent the first and third inequalities result from the triangular inequality for the trace norm while the second and fourth from the non-increase of the trace norm under trace preserving operations \cite{RUSKAI1994}. Taking the limit $n\to\infty$ of both sides of \cref{eq5.15} result in $\lim_{n\to\infty}\norm{\Lambda_n'''(\rho^{\otimes n})-\sigma^{\otimes n}}=0$ so according to \cref{def5.3} $\rho\exasto{O}\sigma$.
\end{proof}
It would be interesting to try and prove that \crefrange{axm4.1}{axm4.8} are satisfied by taking the strong form of the second law as granted. This backwards approach is possible for classical thermodynamics. Due to the additivity of the entropy function all of the Axioms are reduced to relations between real numbers which hold trivially. By construction the unique entanglement measure $E$ is equal to the entanglement cost which is a subadditive function so this approach no longer applies.

A treatment of mixed state entanglement based on \cref{eq5.14} has been applied to two classes of operations \emph{positive partial transpose (PPT) preserving operations} and \emph{asymptotically non entangling operations}.
\subsubsection{PPT preserving operations.}\label{sub5.2.1}
The set of PPT preserving operations is defined by
\begin{defin}\label{def5.3}
A positive trace preserving operation $\Lambda$ is called PPT-preserving if for every density matrix $\sigma$
	\begin{labeq}{eq5.17}
		\Lambda(\sigma^{T_X})^{T_X}\geq0\quad X=A,B
	\end{labeq}
\end{defin}
It has been proven that under PPT-preserving transformations any negatively partially transposed (NPT) state can be distilled into the singlet form \cite{Eggeling2001}. Moreover for the special case of antisymmetric Werner states
\begin{labeq}{eq5.18}
\rho_{\scriptscriptstyle\mathcal{W}}=\frac{I-\Pi_d}{d(d-1)}
\end{labeq}
where $\Pi_d=\sum_{i,j}\ket{ij}\bra{ji}$, entanglement of distillation and entanglement cost are equal \cite{Audenaert2003}. It is not known whether $E_D^{ppt}=E_C^{ppt}$ holds in general for any mixed state.
\subsubsection{Asymptotically non-entangling operations.}\label{sub5.1.2}
A straightforward way of describing a class of operations can be given by formulating mathematically their limitations, for example LOCC operations are a subset of the set of separable operations SEPP which can be defined as those operations that cannot create entangled states from separable ones. Similarly PPT-preserving transformations are defined as those that can never output NPT states when the inputs are PPT. One can go even further and lift the restriction of quantum communication between two parties or even allow them to create a finite amount of entanglement. In  \cite{Plenio2008,Brandao2010} the authors considered the class of \emph{asymptotically non entangling operations} $SEPP(\epsilon)$. Specifically they allowed the possibility that the output of a quantum operation $\Omega$ on a separable state might become entangled, this causes no problems as long as this amount vanishes asymptotically.
\begin{defin}\label{def5.4}
A sequence $\{\Omega_n\}$ is called asymptotically non-entangling if
	\begin{labeq}{eq5.19}
		\lim_{n\to\infty}\epsilon_n=0
	\end{labeq}
where $\epsilon_n=E(\Omega_n(\sigma^{\otimes n}))$ is the amount of entanglement that is created when $\Omega_n$ acts on $n$ copies of a separable state $\sigma$ and $E$ is any entanglement measure.
\end{defin}
The importance of these transformations can be appreciated by the following theorem.
\begin{thrm}\label{thrm5.3}
	For any state $\rho$ 
		\begin{labeq}{eq5.20}
			E_D^{SEPP(\epsilon)}(\rho)=E_C^{SEPP(\epsilon)}(\rho)
		\end{labeq}
\end{thrm}
The most surprising consequence of \cref{thrm5.3} is that it holds true for multipartite states as well. The unique entanglement measure in this case was proven to be equal to the regularized relative entropy of entanglement $E_R^{\infty}$. Because the relation induced by the set of asymptotically entangling operations makes the set of entangled states totally ordered, $SEPP(\epsilon)$ can be considered as describing the natural analogue of adiabatic processes.  The second law of entanglement thermodynamics therefore reads
\begin{stat}[Second law of entanglement thermodynamics under $SEPP(\varepsilon)$.]\label{thrm5.4}
	\begin{labeq}{eq5.21}
		\rho\prec\sigma\quad\text{if and only if}\quad E_R^{\infty}(\rho)\geq E_R^{\infty}(\sigma)
	\end{labeq}
\end{stat}
\section{Conclusions.}\label{chap6}
Although the theory of bipartite entanglement possesses a thermodynamical structure for the special case of pure states, the question is still open for mixed states. The existence of the second law for asymptotically non-entangling operations indicates that this is very likely. Furthermore the current solutions to the problem are more of a theoretic than of practical interest mainly because of the asymptotic number of states required, this is clearly unrealistic for every day experiments in a lab. A detailed examination of mixed state entanglement based on the Axiomatic formulation of thermodynamics is therefore necessary. An important step towards this direction would be the development of necessary and sufficient conditions for transforming any entangled state to another which for pure states and LOCC operations reduces to Nielsen's theorem.

A correspondence between entanglement and thermodynamics if possible could then open the way to using methods from one theory for solving problems in the other. This would surely result in new and exciting phenomena in physics.

\appendix
\addcontentsline{toc}{section}{Appendix.}
\section{Schmidt decomposition.}\label{appa}

A pure bipartite state $\varphi$ belonging to a $\mathbb{C}^{d_A}\otimes\mathbb{C}^{d_B}$ Hilbert space (without loss of generality assume that $d_A\geq d_B$) is given by
\begin{labeq}{app.a1}
	\varphi_{\scriptscriptstyle AB}=\sum_{i=1}^{d_A}\sum_{j=1}^{d_B}c_{ij}\ket{\alpha_i}\ket{\beta_j}
\end{labeq}
where $\left\{\ket{\alpha_i}\right\}_{i=1}^{d_A}$ and $\left\{\ket{\beta_j}\right\}_{j=1}^{d_B}$ form an orthonormal set of bases for Hilbert spaces $\mathbb{C}^{d_A}$ and $\mathbb{C}^{d_B}$ respectively and $c_{ij}\in\mathbb{C}$ such that $\sum_{i,j}|c_{ij}|^2=1$. By treating the complex coefficients in \cref{app.a1} as the components of a  $d_A\times d_B$ dimensional complex matrix $C$ and by applying the \emph{singular value decomposition} theorem there exist unitary matrices $U$ and $V$ such that ${C=U\Sigma V^\dagger}$ with $\Sigma$ a $d_A\times d_B$ dimensional, real, non-negative diagonal matrix. It's components $\sigma_k$ are called the \emph{singular values} of the matrix $C$. Replacing $c_{ij}=\sum_ku_{ik}\sigma_kv_{jk}^\ast$ into \cref{app.a1} transforms the state into.
\begin{labeq}{app.a2}
	\varphi_{\scriptscriptstyle AB}=\sum_{k=1}^{d_A}\sigma_k\(\sum_{i=1}^{d_A}u_{ik}\ket{\alpha_i}\)\(\sum_{j=1}^{d_B}v_{jk}^\ast\ket{\beta_j}\)
\end{labeq}
It is easy to check that $\indexket{k}{A}=\sum_{i=1}^{d_A}u_{ik}\ket{\alpha_i}$ and $\indexket{k}{B}=\sum_{j=1}^{d_B}v_{jk}^\ast\ket{\beta_j}$ form orthonormal sets for $\mathbb{C}^{d_A}$ and $\mathbb{C}^{d_B}$ respectively. 
Setting $\sigma^2_k=\lambda_k$ results in the \emph{Schmidt decomposition} of $\varphi$.
\begin{labeq}{app.a3}
	\varphi_{\scriptscriptstyle AB}=\sum_{i=1}^{d_A}\sqrt{\lambda_i}\ket{ii}
\end{labeq}

\section{Teleportation.}\label{appb}
We shall describe the teleportation protocol for the simple case of a pure qubit $\phi=\alpha\ket{0}+\beta\ket{1}$ where without loss of generality $\alpha$ and $\beta$ are real coefficients with $\alpha^2+\beta^2=1$. Suppose Alice is in possession of the qubit and wishes to teleport it to Bob, to do this they need to share in advance entanglement in the form of a maximally entangled singlet state. The combined state of the qubit and singlet is given by $\phi_{\scriptscriptstyle A'}\otimes\Psi^-_{\scriptscriptstyle AB}$. To show how this is possible it suffices to rewrite Alice's state consisting of the qubit particle and her half of the singlet in the Bell basis.
\begin{align}\label{app.b2}
	\phi_{\scriptscriptstyle A'}\otimes\Psi^-_{\scriptscriptstyle AB}=&\frac{1}{\sqrt{2}}\left[\alpha(\prodket{00}{}{1}{}-\prodket{01}{}{0}{})+\beta(\prodket{10}{}{1}{}-\prodket{11}{}{0}{})\right]_{\scriptscriptstyle A'AB}\nonumber\\
	=&\frac{1}{2}[\alpha(\Phi^++\Phi^-)_{\scriptscriptstyle A'A}\otimes\ket{1}_{\scriptscriptstyle B}-\alpha(\Psi^++\Psi^-)_{\scriptscriptstyle A'A}\otimes\ket{0}_{\scriptscriptstyle B}\nonumber\\
	&+\beta(\Psi^+-\Psi^-)_{\scriptscriptstyle A'A}\otimes\ket{1}_{\scriptscriptstyle B}-\beta(\Phi^++\Phi^-)_{\scriptscriptstyle A'A}\otimes\ket{0}_{\scriptscriptstyle B}]\nonumber\\
	=&\frac{1}{2}[\Phi^+_{\scriptscriptstyle A'A}\otimes(\alpha\ket{1}+\beta\ket{0})_{\scriptscriptstyle B}+\Phi^-_{\scriptscriptstyle A'A}\otimes(\alpha\ket{1}-\beta\ket{0})_{\scriptscriptstyle B}\nonumber\\
	&+\Psi^+_{\scriptscriptstyle A'A}\otimes(-\alpha\ket{0}+\beta\ket{1})_{\scriptscriptstyle B}+\Psi^-_{\scriptscriptstyle A'A}\otimes(-\alpha\ket{0}-\beta\ket{1})_{\scriptscriptstyle B}]
\end{align}
All Alice has to do now is perform a von-Neumann measurement to determine the Bell state of her pair and then inform Bob about the result of her measurement using a classical channel \footnote{Note that in order to do so assuming the channel is noiseless she has to use two classical bits.}. Depending on the message he received Bob can simply apply a suitable unitary rotation on his particle which will always result in him obtaining the desired state $\phi_{\scriptscriptstyle B}$. 

More precisely if Alice finds that her pair is in a $\Psi^-$ state then one look at \cref{app.b2} is enough to show that there is no need for Bob to rotate his particle since apart from an overall phase it's state is exactly the one he wished to receive. For a $\Phi^+$ result he must apply a $\sigma_x$ Pauli rotation which flips the basis vectors, similarly for $\Phi^-$ he must apply $\sigma_z\sigma_x$ and for $\Psi^+$ $\sigma_z$. This method is readily generalized to the case where it is desired to teleport a q-dit, in this case it is necessary to share a $\Psi_d^+$ state.

\addcontentsline{toc}{section}{References.}
\bibliographystyle{ieeetr}
\bibliography{Thesis}
\end{document}